\documentclass[journal]{IEEEtran}
\IEEEoverridecommandlockouts
\usepackage{cite}
\usepackage{amsmath,amssymb,amsfonts,amsthm}
\usepackage{algorithmic}
\usepackage{graphicx}
\usepackage{textcomp}
\usepackage{xcolor}
\usepackage{mathtools}
\usepackage{bbm}
\usepackage[short]{optidef}

\newtheorem{theorem}{Theorem}
\newtheorem{proposition}{Proposition}

\newtheorem{proof*}{Proof}

\DeclareMathOperator*{\argmin}{argmin}
\def\BibTeX{{\rm B\kern-.05em{\sc i\kern-.025em b}\kern-.08em
    T\kern-.1667em\lower.7ex\hbox{E}\kern-.125emX}}

\expandafter\def\expandafter\normalsize\expandafter{%
\normalsize%
\setlength\abovedisplayskip{4.5pt}%
\setlength\belowdisplayskip{4.5pt}%
}

\allowdisplaybreaks
    
\begin{document}

\title{Large-scale Tunable Liquid Lens-assisted VLC Systems under Random Receiver Orientation}

\author{Kapila~W.~S.~Palitharathna,~\IEEEmembership{Member,~IEEE,} Constantinos~Psomas,~\IEEEmembership{Senior~Member,~IEEE,} \\Gaofeng~Pan,~\IEEEmembership{Senior~Member,~IEEE,}~and~Ioannis~Krikidis,~\IEEEmembership{Fellow,~IEEE}\vspace{-6mm}
\thanks{K. W. S. Palitharathna and I. Krikidis are with the IRIDA Research Centre for Communication Technologies, Department of Electrical and Computer Engineering, University of Cyprus, 1678 Nicosia, Cyprus (e-mails: \{palitharathna.kapila, krikidis\}@ucy.ac.cy).}
\thanks{C. Psomas is with the Department of Computer Science and Engineering, European University Cyprus, Nicosia, Cyprus (email: c.psomas@euc.ac.cy).}
\thanks{G. Pan is with School of Cyberspace Science and Technology, Beijing Institute of Technology, Beijing 100081, China (e-mail: gfpan@bit.edu.cn).}}

\maketitle

\begin{abstract}
This paper investigates the performance of tunable liquid lens (TLL)-assisted receivers in large-scale visible light communication (VLC) systems under random receiver orientation. A simple electrowetting-based TLL architecture is proposed, capable of dynamically steering the incident optical signal toward the photodiode receiver by adjusting the orientation of the liquid interface. The proposed architecture enhances the desired signal reception while mitigating interference from neighboring access points (APs). The spatial distribution of APs is modeled using a Matérn hard-core point process, whereas receiver orientation is characterized by uniformly distributed azimuth angles and Gaussian-distributed polar angles. Furthermore, a tractable mathematical optical channel model is developed to capture the combined effects of AP/receiver locations, receiver orientation, and lens adjustment angles on the VLC channel gain. Based on this framework, three lens orientation strategies, namely best signal reception (BSR), closest LED selection, and vertical upward lens orientation, are proposed to improve system performance under dynamic receiver conditions. Using stochastic geometry tools, exact and approximate analytical expressions for the outage probability are derived for each scheme. Numerical results verify the accuracy of the developed analysis and demonstrate that the proposed TLL-assisted receiver architecture significantly improves the robustness of VLC systems under severe receiver orientation fluctuations and dense AP deployments. In particular, the BSR scheme reduces the outage probability by $57.1\%$ compared with conventional fixed-lens receivers at an AP height of $3.5$ m and AP density of $0.2~\text{m}^{-2}$. The presented analytical framework and numerical results provide useful design insights for the deployment of future TLL-assisted VLC networks.\vspace{-1mm}
\end{abstract}

\begin{IEEEkeywords}
Visible light communication, tunable liquid lens, stochastic geometry, random receiver orientation.
\end{IEEEkeywords}\vspace{-1mm}

\section{Introduction}\vspace{-1mm}
Next-generation wireless networks are expected to provide high data rates, ultra-low latency, enhanced security, and massive connectivity to support emerging applications such as augmented reality, e-health, and Industry 4.0. To this end, visible light communication (VLC) has emerged as a promising technology for future wireless networks, particularly for indoor short-range communication systems. In particular, VLC exploits the unlicensed visible light spectrum (380 nm to 780 nm), which offers significantly larger bandwidth than conventional radio frequency (RF) systems to achieve high data rates~\cite{Ghassemlooy}. Further, VLC systems are capable of simultaneously providing energy-efficient illumination and high-speed wireless communication by using low-cost light-emitting diodes (LEDs) as transmitters and photodiodes (PDs) as receivers~\cite{Ghassemlooy}.

Despite these advantages, VLC systems require a strong line-of-sight (LoS) link for efficient information transfer. The availability of a reliable LoS channel strongly depends on transmitter-receiver alignment, blockages caused by static objects or humans movement, and user mobility. Several studies have investigated the performance degradation caused by blockage and receiver dynamics, and proposed conditions and techniques to improve system robustness~\cite{Srivastava_2021, Singh_2022, Beysens_2020,Abumarshoud_2022, Kapila_2022}. In~\cite{Srivastava_2021}, the downlink performance of indoor VLC systems under static and mobile human blockages was analyzed. A user-guidance framework to mitigate blockage effects was proposed in~\cite{Singh_2022}. In~\cite{Beysens_2020}, mirror-based intelligent reflective surfaces (IRSs) were utilized to improve the sum rate of indoor VLC systems, while authors in~\cite{Abumarshoud_2022} investigated IRS-assisted non-orthogonal multiple access VLC systems to improve link reliability. In addition, machine learning-based performance enhancement techniques under random receiver orientation and mobility conditions were studied in~\cite{Kapila_2022}. 

Recently, optical lens technology have introduced several tunable liquid lens (TLL) architectures that can dynamically adjust optical properties to improve signal reception quality~\cite{Ndjiongue_2021, Ngatched_20212, Cheng_2021, Lee_2019, Zohrabi_2016}. Among these technologies, liquid crystal-based structures capable of dynamically tuning the refractive index to manipulate the light propagation direction have been examined within VLC systems~\cite{Ndjiongue_2021, Ngatched_20212}. In~\cite{Ndjiongue_2021} and~\cite{Ngatched_20212}, a liquid crystal-based IRS is utilized at the receiver to dynamically steer light beams toward the effective area of the PD, thereby optimizing signal reception and improving overall communication performance. Although omitted in the context of VLC systems, several new non-mechanical liquid lens architectures have been proposed that can change the orientation and shape of the liquid surface and hence control the light propagation direction~\cite{Cheng_2021,Lee_2019,Zohrabi_2016}. Numerous non-mechanical electrowetting surface-based liquid lens architectures have been proposed in~\cite{Cheng_2021}. The authors in~\cite{Lee_2019} introduced an innovative three-dimensional beam steering methodology that leverages an electrowetting-based liquid lens in conjunction with a liquid prism, enhancing the precision of light manipulation. In~\cite{Zohrabi_2016}, the authors presented techniques for one- and two-dimensional beam steering employing multiple tunable liquid lenses, demonstrating significant advancements in beam control capabilities. In addition, recent research has explored the integration of liquid lens systems with mechanical structures to achieve enhanced dynamic beam steering functionalities~\cite{Tian_2022,Kapila_2025}. In~\cite{Tian_2022}, an adaptable liquid lens is studied which has three degrees of freedom i.e., focal length, azimuth angle, and polar angle. The adjustment of focal length is facilitated by the application of a vertical mechanical force on a ring positioned around the liquid, while a mechanical framework employs magnetic forces to enable tilting of the ring, allowing for precise modulation of both azimuth and polar angles~\cite{Tian_2022}. Furthermore, a mechanical liquid lens architecture has been used to improve the communication efficiency by reducing the interference among channels in multiple-input multiple-output (MIMO) VLC systems~\cite{Kapila_2025}. However, channel models in such systems are mathematically not tractable and require exhaustive simulations. Therefore, there is a need for simple TLL architectures that facilitate mathematically tractable channel modeling and analytical performance evaluation, while also enhancing system performance. To this end, our previous work in~\cite{Kapila_20252}, introduced a simple cuboid TLL model using electrowetting surfaces and provided a mathematical tractable VLC channel model with the proposed TLL-based receiver.

On the other hand, the large-scale analysis of VLC systems is required to obtain the average performance of proposed systems. In particular, stochastic geometry (SG) has been used as a powerful and tractable mathematical tool for analyzing the impact of key parameters on the network performance\cite{Singh_2021,Lutz_2021,Liang_2018,Singh_2022, Liang_2018_2}. In~\cite{Singh_2022}, SG is used to analyze the delay outage rate and the information outage rate of a vehicular-VLC network. In~\cite{Singh_2021}, SG is used to characterize the interference in RF and VLC-based VLC systems. In~\cite{Lutz_2021}, user mobility analysis for RF/VLC hybrid network is conducted by deriving the user-to-base station association probabilities and handover rates. The secrecy outage probability and the ergodic secrecy rates are derived for a 3D multiuser VLC system in~\cite{Liang_2018}. In~\cite{Liang_2018_2}, a new mathematical framework to analyze the coverage probability of multi-user VLC networks has been presented taking into account the idle probability of access points (APs) that are not associated with any users. The work in~\cite{Hina_2018}, provides an SG framework to perform the coverage and rate analysis of a typical user in co-existing VLC/RF networks. In~\cite{Kong_2019}, the RF/VLC base station intensities are optimized to minimize the area power consumption under outage probability constraint. More recently, a large scale analysis of TLL-based VLC systems is presented in~\cite{Christodoulos_2026}, which accounts for the user mobility and blockages. However, to the best of the authors' knowledge this is the first work that presents a TLL architecture and models its channel under random receiver orientation conditions, and analyzes the outage performance for different lens orientations. 

In this paper, we investigate a large-scale indoor VLC system in which multiple APs are distributed according to a Matérn hard-core point process (MHCPP), and each randomly deployed user is equipped with an electrowetting liquid-lens to improve the signal reception at the PD receiver. 
Specifically, the effects of random receiver orientation are incorporated into the system and channel models, and the TLL orientation is dynamically adjusted based on the receiver location, receiver orientation, and the location of the closest AP. To model the channel gains of VLC links, we present an accurate mathematical model to capture the channel gain variation for different receiver positions, receiver orientation, and lens orientation. We present several lens orientation schemes with different system complexity. The outage performance of each lens orientation scheme is analyzed and closed-form approximate expressions for the outage probability of each scheme is presented. The main contributions of this paper are:
\begin{itemize}
    \item An indoor VLC network is considered with APs distributed according to an MHCPP-Type II and mobile receivers with random receiver orientations. A simple electrowetting surface-based TLL architecture is proposed to improve signal reception at a receiver from the \textit{tagged AP} while suppressing interference from neighboring APs. The proposed architecture enables adaptive control of the lens orientation independently of the receiver orientation, allowing incident optical signals to be refracted at the lens surface and focused onto the PD. As a result, communication performance is enhanced under user mobility and random receiver orientation.
    \item We develop a tractable mathematical framework to characterize the VLC channel gains under different receiver locations, receiver orientations, and lens orientations. Based on this framework, three lens orientation schemes, namely best signal reception (BSR), closest LED selection (CLS), and vertical upward lens orientation (VULO), are proposed to improve the received signal quality and overall system performance. In particular, the BSR scheme aims to maximize the received signal power from the \textit{tagged AP}, whereas the CLS and VULO schemes offer lower-complexity alternatives with reduced lens-control overhead.
    \item Assuming that the AP locations follow an MHCPP-Type II, we analyze the outage probability of the proposed TLL orientation schemes using SG tools. In particular, we derive the Laplace transform of the aggregate interference under each scheme and obtain approximate closed-form expressions for the outage probability using the developed channel model. Numerical results validate the analytical framework and demonstrate significant performance gains over conventional fixed-lens and no-lens VLC receivers across a wide range of system parameters, including receiver orientation, AP height, and AP density. For example, at an AP height of $3.5$ m and AP density of $0.2~\text{m}^{-2}$, the proposed BSR scheme reduces the outage probability from $1.4\times10^{-1}$ to $6\times10^{-2}$ compared to a no-lens receiver.
\end{itemize}

The remainder of this paper is organized as follows. In Section \ref{system}, the system model, the associated liquid lens architecture, and the receiver orientation model are presented. Section \ref{channel} characterizes the liquid lens-assisted VLC channel gains and Section \ref{lens_orientation} presents the proposed lens orientation schemes. Section \ref{outage} presents the large-scale outage probability analysis of the presented lens orientation schemes. Numerical results for various system parameters are presented in Section \ref{results}. Finally, conclusions are drawn in Section \ref{conclusion}.

\emph{Notation:} $\mathbb{R}^n$ denotes the $n$-dimensional Euclidean space, $||x||$ denotes the Euclidean norm of $x \in \mathbb{R}^n$, $\mathbb{P}[X]$ denotes the
probability of the event $X$, $\mathbb{E}[X]$ represent the expected value of $X$, $\mathbbm{1}(\cdot)$ denotes the indicator function, $|\mathcal{A}_{\mathbf{y}}|$ denotes the Lebesgue measure (i.e., area) of the set $\mathcal{A}_{\mathbf{y}}$, $\mathcal{U}[a,b]$ denotes the continuous uniform distribution over the interval $[a,b]$, $\mathcal{N}(\mu,\sigma^2)$ denotes a Gaussian distribution with mean $\mu$ and variance $\sigma^2$, $\text{Im} [x]$ represents the imaginary component of a complex number $x$, $\mathcal{B}(x,R)$ denotes a ball of radius $R$ centered at a point $x$, $I_0(\cdot)$ denotes the modified Bessel function of first kind, $\Gamma(\cdot)$ denotes the Gamma function, $\Gamma(:,:)$ denotes the upper incomplete Gamma function, and $E_1(x) = \int_x^\infty \frac{e^{-t}}{t}\,dt$ denotes the exponential integral function.

\section{System Model}\label{system}
    \begin{figure}[!t]
        \centering
        \includegraphics[width=0.95\columnwidth]{"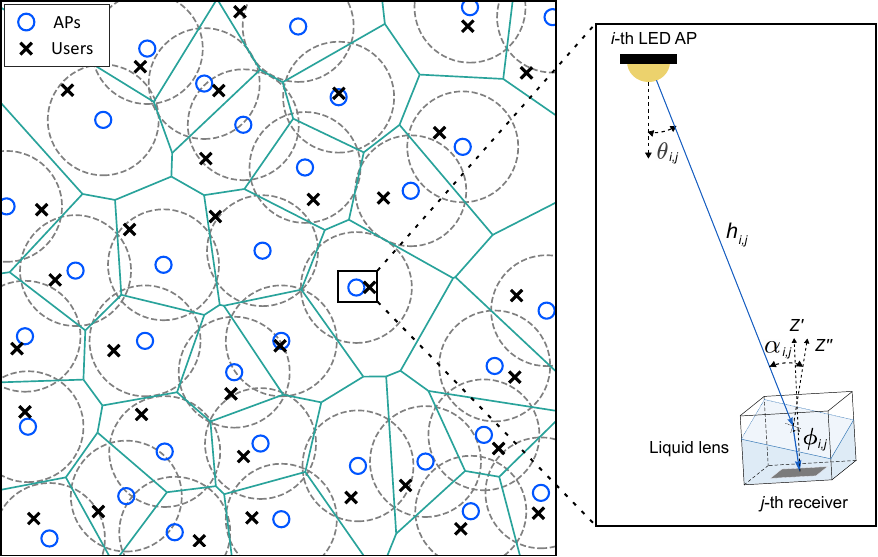"}
        \vspace{-1mm}
        \caption{The Voronoi tessellation of a large-scale TLL-assisted VLC network, where APs and receivers are represented by circles and crosses, respectively.}
        \label{largescale}
        \vspace{-3mm}
    \end{figure}

\subsection{Network Topology}
We consider an indoor VLC system composed of ceiling-mounted LEDs that serve as the APs, and users equipped with electrowetting TLL-assisted VLC receivers. The locations of the APs are modeled as points of a MHCPP-Type II\footnote{Although LEDs are often deployed regularly within small rooms, large-scale VLC networks may exhibit spatial irregularities due to practical deployment constraints. The MHCPP captures such spatial repulsion through a minimum AP separation distance while maintaining analytical tractability, making it a suitable model for infrastructure deployments and has been shown to closely match practical deployment datasets~\cite{Lutz_2021}.}, denoted by $\Phi_m$, which is generated by a dependent thinning of a stationary Poisson point process (PPP), denoted by $\Phi_p = \{x_i\in \mathbb{R}^2, i\in\mathbb{N}^+\}$, with spatial density $\lambda_p$ APs/$\text{m}^2$~\cite{Ibrahim_2013}. Hence, $\Phi_m = \{y:y\in \Phi_p,m_y<m_x, \forall x\in (\Phi_p\cap \mathcal{B}(y,R))\setminus \{y\} \}$, where $m_x\sim\mathcal{U}[0,1]$ is an independent random mark assigned to each point $x\in \Phi_p$, and $R$ denotes the minimum hard-core separation distance between neighboring APs. The density of the thinned process $\Phi_m$ is given by $\lambda_m = (1-\exp(-\lambda_p\pi R^2))/\pi R^2$~\cite{Yeom_2023}. Furthermore, the receiver locations follow an arbitrary independent point process $\Psi_u$ with spatial density $\lambda_u \gg \lambda_m$. We assume that all APs are equipped with downward-oriented LEDs, while all user devices are equipped with a PD and a TLL. For the multiple access scheme, we employ an orthogonal multiple access technique, e.g., time division multiple access, such that each AP serves a single user at a time, thereby eliminating intra-cell interference. Without loss of generality, and following Slivnyak's theorem~\cite{Haenggi_2012}, the analysis focuses on a typical user located at the origin, with the results holding for all users in the network.

For modeling purposes, we define the room coordinate system, the $j$-th receiver local coordinate frame, and the $j$-th lens surface coordinate frame as $OXYZ$, $O'_jX'_jY'_jZ'_j$, and $O''_jX''_jY''_jZ''_j$, respectively. Due to practical conditions such as user mobility and random receiver orientation, the $j$-th receiver is subjected to an azimuth angle of $\theta_{R,j}$ and a polar angle $\phi_{R,j}$~\cite{soltani_2019}. To align with empirical measurements of mobile users, we model these angles as $\theta_{R,j}\sim \mathcal{U}(0,2\pi)$, and $\phi_{R,j}\sim \mathcal{N}(\mu_{\phi_R},\sigma_{\phi_R}^2)$~\cite{soltani_2019}. In addition, the position of the $i$-th AP is denoted by $\hat{\boldsymbol{P}}_{T,i} = [x_{T,i}, y_{T,i}, z_{T,i}]^\top$, and the position of the $j$-th receiver is denoted by $\hat{\boldsymbol{P}}_{R,j} = [x_{R,j}, y_{R,j}, z_{R,j}]^\top$. Throughout this work, the receiver orientation is random and independent of the lens orientation, while the lens orientation is adaptively controlled according to the adopted TLL scheme.

\subsection{Liquid Lens Model}\label{liquid lens}
    \begin{figure}[!t]
        \centering
        \includegraphics[width=0.95\columnwidth]{"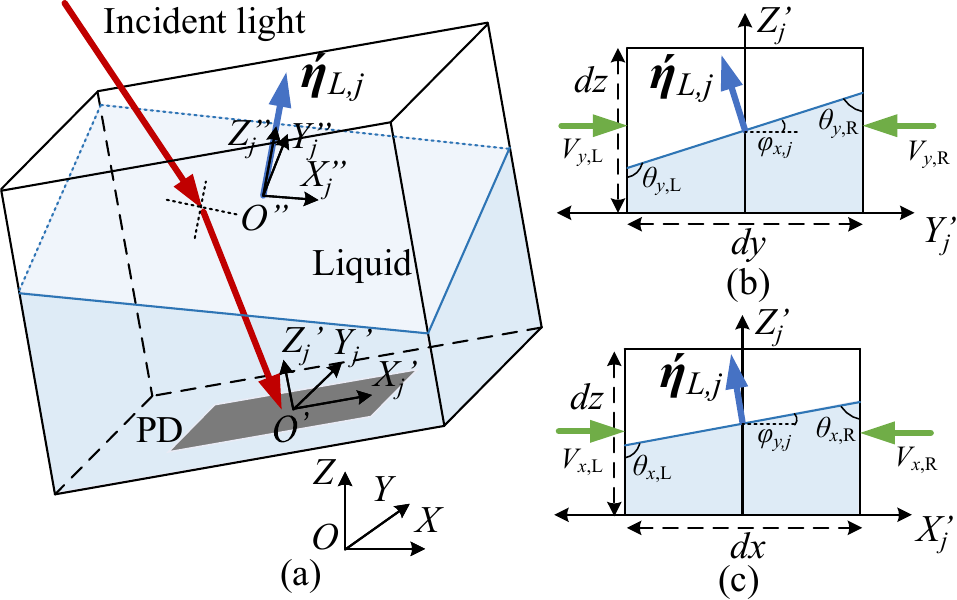"}\
        \vspace{-1mm}
        \caption{Electrowetting lens-assisted VLC receiver with random receiver orientation. (a) Perspective view. (b) $X_j'$-axis view. (c) $Y_j'$-axis view.}
        \label{fig:f2}
        \vspace{-2mm}
    \end{figure}

As illustrated in Fig. 2(a), the proposed electrowetting TLL is a cuboid structure with dimensions $d_x\times d_y \times d_z$, partially filled with an optically transparent liquid. The four side surfaces of the lens are composed of electrowetting surfaces~\cite{Lee_2013,Cheng_2021}. By varying  the voltages applied to these surfaces, namely, $V_{x,L}$, $V_{x,R}$, $V_{y,L}$, and $V_{y,R}$, the corresponding contact angles at the walls, $\theta^L_{x,L}$, $\theta^L_{x,R}$, $\theta^L_{y,L}$, and $\theta^L_{y,R}$, can be dynamically adjusted~\cite{Lee_2013}. This variation in contact angles modifies the normal vector of the liquid surface, denoted as $\hat{\boldsymbol{\eta}}_{L,j}$, which can be controlled through the applied voltages. Fig. 2(b) and Fig. 2(c) present side views of the TLL along the $X'_j$- and $Y'_j$- directions, respectively. Assuming an initial contact angle of $\frac{\pi}{2}$, the relationship between the applied voltages and the contact angle can be expressed as~\cite{Lee_2013}
\begin{align}\label{equ:e1}
\hspace{-1mm}
    \theta^L_{M,N} = \cos^{-1}{\bigg(\frac{kV_{M,N}^2}{d_{M}}\bigg)}, \forall M\in\left\{x,y\right\}, \forall  N\in \left\{L,R\right\}, 
\end{align}
where $k = \frac{\varepsilon_0\varepsilon_1}{2\gamma_{L}}$, with $\varepsilon_1$ representing the relative permittivity of the dielectric layer,
$\varepsilon_0$ denoting the permittivity of vacuum, and $\gamma_{L}$ being the interfacial surface tension between the liquid and vapor phases~\cite{Cheng_2021}. Here, the applied voltages satisfy $-\left(\frac{d_M}{k}\right)^{1/2}\le V_{M,N}\le \left(\frac{d_M}{k}\right)^{1/2}$ to ensure physically realizable contact angles. Furthermore, assuming the liquid surface remains flat, the relationship between the voltages applied to two opposing electrowetting surfaces can be expressed as
\begin{align}\label{equ:e2}
    \sum\nolimits_{N \in \{L,R\}}\cos^{-1}\bigg(\frac{kV_{M,N}^2}{d_{M}}\bigg) = \pi, \quad \quad \forall M\in \{x,y\}.    
\end{align}
Using basic trigonometry, the angles of the liquid surface of the $j$-th receiver with respect to the $X'_jY'_j$ plane on the $X'_jZ'_j$ and $Y'_jZ'_j$ planes are $\psi_{M,j} = \frac{\pi}{2}-\theta^L_{M,R}, \forall M\in \{x,y\}$. 

\subsection{Association and SINR Model}
We adopt the nearest-AP association criterion due to its simplicity in implementation and analytical tractability. Consider a typical user located at position $\mathbf{y} = \{x_{R,j}, y_{R,j}\}$ on the floor plane is associated with the $i$-th AP located at $\mathbf{x}_i=\{x_{T,i},y_{T,i}\}$ on the LED plane. The corresponding signal-to-interference-plus-noise ratio (SINR) is expressed as
\begin{align}\label{SINR1}
        \text{SINR}_{i} &= \frac{\textit{l}(\mathbf{x}_i,\mathbf{y})}{I(\Phi_m\backslash\mathbf{x}_i)+\sigma^2},
    \end{align}
where $\textit{l}(\mathbf{x}_i,\mathbf{y}) = r_{PD}^2P_i^2k_1^2\text{c}^2(\phi_{i,j})/(||\mathbf{x}_i-\mathbf{y}||^2+h_u^2)^{m+2}$ is the signal power, $k_1 = \frac{(m+1)Ah_u^m}{2\pi}$ is a constant, $m=-\ln(2)/\ln(\cos(\theta_{1/2}))$ represents the Lambertian order of the LED, $A$ is the aperture area of the PD, $r_{PD}$ is the responsivity of the PD, $P_i$ denotes the transmit power, $h_u = z_{T,i}-z_{R,j}$ is the height between the receiver plane and the LED plane, and $I(.)$ represents the aggregate sum interference given by
\begin{align}
    I(\Phi\backslash\mathbf{x}_i) = \sum\nolimits_{\mathbf{x}_k\in\Phi_m\backslash\mathbf{x}_i} \textit{l}(\mathbf{x}_k,\mathbf{y}).
\end{align}
For analytical convenience, we assume that the index of the serving AP is 0, such that $\mathbf{x}_0 = \argmin_{\mathbf{x}_i} ||\mathbf{x}_i-\mathbf{y}||$. Accordingly, the SINR at this receiver is given by 
\begin{align}\label{SINR2}
        \text{SINR}_{0}(\mathbf{y}) \hspace{-0.5mm}= \hspace{-0.5mm}\frac{\text{c}^2(\phi_{0,j})(||\mathbf{x}_0-\mathbf{y}||^2+h_u^2)^{-(m+2)}}{\hspace{-2.5mm}\displaystyle\hspace{-0.5mm}\sum_{\mathbf{x}_k\in\Phi_m\backslash\mathbf{x}_0}\hspace{-4.5mm}\text{c}^2(\phi_{k,j})(||\mathbf{x}_k-\mathbf{y}||^2\hspace{-0.9mm}+\hspace{-0.9mm}h_u^2)^{-(m\hspace{-0.5mm}+\hspace{-0.5mm}2)}\hspace{-0.9mm}+\hspace{-0.9mm}\sigma'^2},
    \end{align}
where $\sigma'^2 = \sigma^2/(r_{PD}^2P_i^2k_1^2)$. In the subsequent sections, we develop analytical framework to evaluate the performance of this typical user under different lens orientation schemes.

In indoor VLC systems, the closest AP typically dominates the received SINR due to the highly directional and rapidly decaying nature of optical propagation and the absence of small-scale fading. Although certain lens orientation configurations may alter the dominant received link through angular gain variations, the nearest-AP association remains a highly accurate approximation to strongest-SINR association in practical VLC deployments. Note that all our lens orientations are proposed such that there is no violation to this condition.
Therefore, for analytical tractability, the nearest-AP association model is approximated using the equivalent strongest-SINR coverage formulation commonly adopted in SG analyses~\cite{Gupta_2018, Liang_2018_2}.

\section{VLC Channel Model}\label{channel}
In this section, we present the light propagation model adopted in our VLC system. The optical channel gain $h_{i,j}$ between the $i$-th LED and the $j$-th PD is characterized using the Lambertian radiation model, and is given by~\cite{soltani_2019}
	\begin{equation}\label{equ:e10}
		h_{i,j}	=
		\displaystyle\frac{(m+1)A}{2\pi d_{i,j}^2}\cos^m(\theta_{i,j})\cos(\phi_{i,j})\Pi\bigg(\frac{\phi_{i,j}}{\phi_{FoV}}\bigg),
	\end{equation}
where $d_{i,j}$ is the Euclidean distance between the LED and the receiver, $\theta_{1/2}$ is the half-power semi-angle of the LED, $\theta_{i,j}$ is the irradiance angle of the LED, $\phi_{i,j}$ is the incident angle at the PD, $\phi_{FoV}$ is the field-of-view (FoV) of the PD, and $\Pi(x)$ is the rectangular function, equal to $1$ for $|x|\le 1$ and $0$ otherwise.

\subsection{Refraction of Light from the TLL}
We derive analytical expressions for the channel gain $h_{i,j}$ of the VLC link, incorporating the proposed TLL. To achieve this, we compute the relevant rotation matrices, unit normal vectors, and the direction of light refraction.
The coordinate frame of the $j$-th receiver, denoted as $O'_jX'_jY'_jZ'_j$, is modeled as a rotated version of the room's global coordinate frame $OXYZ$. This transformation consists of a rotation by $-\theta_{R,j}$ around the global $z$-axis, followed by a rotation by $-\phi_{R,j}$ around the intermediate $y$-axis. The resulting composite rotation is described by the rotation matrix $^0\mathbf{R}_1(\theta_{R,j},\phi_{R,j})=\mathbf{R}_y(-\phi_{R,j})\mathbf{R}_z(-\theta_{R,j})$, where $\mathbf{R}_y(-\phi_{R,j})$ and $\mathbf{R}_z(-\theta_{R,j})$ denote the rotation matrices about the $y$-axis and $z$-axis, respectively. The transformation simplifies to 
    \begin{equation}\label{equ:rot01}
        \hspace{-2.1mm}^0\mathbf{R}_1(\theta_{R,j},\phi_{R,j})
        = 
        \begin{bmatrix}
         \text{c}{\theta_{R,j}}\text{c}{\phi_{R,j}} & \text{s}{\theta_{R,j}}\text{c}{\phi_{R,j}} & -\text{s}{\phi_{R,j}}\\
         -\text{s}{\theta_{R,j}} & \text{c}{\theta_{R,j}}& 0\\
        \text{c}{\theta_{R,j}}\text{s}{\phi_{R,j}} & \text{s}{\theta_{R,j}}\text{s}{\phi_{R,j}} & \text{c}{\phi_{R,j}}\\
        \end{bmatrix},
    \end{equation}
where $\cos{\theta}$, and $\sin{\theta}$ are denoted as $\text{c}\theta$ and $\text{s}\theta$, respectively. The normal vector to the receiver is aligned along the $Z'_j$ axis. Therefore, using passive rotation between two coordinate frames, the unit normal vector of the receiver in the room's coordinate frame is $\hat{\boldsymbol{\eta}}_{R,j}={^0\mathbf{R}_1^{-1}(\theta_{R,j},\phi_{R,j})[0\quad 0 \quad 1]^T}$. Expanding this expression, $\hat{\boldsymbol{\eta}}_{R,j}$ can be expressed as
    \begin{equation}\label{equ:rot011}
        \hat{\boldsymbol{\eta}}_{R,j}
        = 
        \begin{bmatrix}
        \text{c}{\theta_{R,j}}\text{s}{\phi_{R,j}}, \quad 
        \text{s}{\theta_{R,j}}\text{s}{\phi_{R,j}}, \quad 
        \text{c}{\phi_{R,j}}
        \end{bmatrix}^T.
    \end{equation}
The transformation between the $j$-th receiver's coordinate frame $O'_jX'_jY'_jZ'_j$ and the coordinate frame at the surface of the TLL, $O''_jX''_jY''_jZ''_j$, is represented by the rotation matrix $^1\mathbf{R}_2(\psi_{x,j},\psi_{y,j})=\mathbf{R}_y(\psi_{y,j})\mathbf{R}_x(\psi_{x,j})$, where $\mathbf{R}_y(\psi_{y,j})$ and $\mathbf{R}_x(\psi_{x,j})$ denote rotations about $y$- and $x$- axis by angles $\psi_{y,j}$ and $\psi_{x,j}$, respectively. The combined rotation matrix is
\begin{equation}\label{equ:rot0112}
    \hspace{-2mm}^1\mathbf{R}_2(\psi_{x,j},\psi_{y,j})
    = 
    \begin{bmatrix}
     \text{c}{\psi_{y,j}} & \text{s}{\psi_{y,j}}\text{s}{\psi_{x,j}} & \text{s}{\psi_{y,j}}\text{c}{\psi_{x,j}}\\
     0 & \text{c}{\psi_{x,j}}& -\text{s}{\psi_{x,j}}\\
    -\text{s}{\psi_{y,j}} & \text{c}{\psi_{y,j}}\text{s}{\psi_{x,j}} & \text{c}{\psi_{y,j}}\text{c}{\psi_{x,j}}\\
    \end{bmatrix}.
\end{equation}
The unit normal vector to the lens surface can be considered as a unit vector aligned with its local $z''_j$-axis. Expressed in the room's global coordinate frame, this vector is given  by $\hat{\boldsymbol{\eta}}_{L,j}= {^0\mathbf{R}_1^{-1}(-\theta_{R,j},-\phi_{R,j})}^1\mathbf{R}_2^{-1}(\psi_{x,j},\psi_{y,j})[0\quad 0 \quad 1]^T$. This expression can be further simplified as
{\small\begin{equation}\label{fun:len}
    \hat{\boldsymbol{\eta}}_{L,j} \hspace{-1mm}= \hspace{-1mm}
    \begin{bmatrix}
     \text{c}{\psi_{y,j}}(\text{c}{\theta_{R,j}}\text{s}{\phi_{R,j}}\text{c}{\psi_{x,j}}\hspace{-0.5ex}-\text{s}{\theta_{R,j}}\text{s}{\psi_{x,j}})-\text{c}{\theta_{R,j}}\text{c}{\phi_{R,j}}\text{s}{\psi_{y,j}}\\
     \text{c}{\psi_{y,j}}(\text{s}{\theta_{R,j}}\text{s}{\phi_{R,j}}\text{c}{\psi_{x,j}}+\text{c}{\theta_{R,j}}\text{s}{\psi_{x,j}})-\text{s}{\theta_{R,j}}\text{c}{\phi_{R,j}}\text{s}{\psi_{y,j}}\\
     \text{c}{\phi_{R,j}}\text{c}{\psi_{y,j}}\text{c}{\psi_{x,j}}+\text{s}{\phi_{R,j}}\text{s}{\psi_{y,j}}
     \end{bmatrix}.
\end{equation}}

The direction vector along the receiver-transmitter trajectory is denoted as $\hat{\boldsymbol{\eta}}_{RT,i,j} = [e_{RT,x,i,j}, e_{RT,y,i,j}, e_{RT,z,i,j}]^T = \frac{\hat{\boldsymbol{P}}_{T,i}-\hat{\boldsymbol{P}}_{R,j}}{||\hat{\boldsymbol{P}}_{T,i}-\hat{\boldsymbol{P}}_{R,j}||}$, where $e_{RT,x,i,j}$, $e_{RT,y,i,j}$, and $e_{RT,z,i,j}$ denote the components along $x$-, $y$-, and $z$-axes, respectively. The light refraction direction vector, $\hat{\boldsymbol{\eta}}_{ref,i,j}$, can be derived using the vector form of Snell's law, and can be expressed as~\cite{Kapila_2025} 
    \begin{align}\label{equ:ref_o}
        &\hat{\boldsymbol{\eta}}_{ref,i,j} = n_l^{-1}\big[\hat{\boldsymbol{\eta}}_{L,j}\times (\hat{\boldsymbol{\eta}}_{L,j}\times \hat{\boldsymbol{\eta}}_{RT,i,j}) \nonumber\\
        &- \hat{\boldsymbol{\eta}}_{L,j}\left(n_l^2-(\hat{\boldsymbol{\eta}}_{L,j}\times\hat{\boldsymbol{\eta}}_{RT,i,j})\cdot(\hat{\boldsymbol{\eta}}_{L,j}\times\hat{\boldsymbol{\eta}}_{RT,i,j})\right)^{\frac{1}{2}}\big],\nonumber \\
        &=\hspace{-0.5mm}{n_{l}}^{-1}\big(\hat{\boldsymbol{\eta}}_{L,j}\big(\text{c}(\alpha_{i,j})-(n_l^2 \hspace{-0.5mm}- \hspace{-0.5mm}\text{s}^2(\alpha_{i,j}))^{\frac{1}{2}}\big)\hspace{-0.5mm}-\hspace{-0.5mm}\hat{\boldsymbol{\eta}}_{RT,i,j}\big),
    \end{align}
where $n_l$ is the relative refractive index of the liquid. This simplification is based on the vector identity, $\boldsymbol{A}\times (\boldsymbol{A}\times \boldsymbol{B}) = (\boldsymbol{A}\cdot\boldsymbol{B})\cdot\boldsymbol{A} - (\boldsymbol{A}\cdot\boldsymbol{A})\cdot \boldsymbol{B}$ and the relation $\text{s}^2(\alpha_{i,j}) =(\hat{\boldsymbol{\eta}}_{L,j}\times\hat{\boldsymbol{\eta}}_{RT,i,j})\cdot(\hat{\boldsymbol{\eta}}_{L,j}\times\hat{\boldsymbol{\eta}}_{RT,i,j})$, where $\alpha_{i,j} = \cos^{-1}{\left(\hat{\boldsymbol{\eta}}_{L,j}\cdot\hat{\boldsymbol{\eta}}_{RT,i,j}\right)}$ is the light incident angle at the liquid surface.

\subsection{Simplified Channel Model}
We derive a simplified expression for $h_{i,j}$ as a function of the receiver position, receiver orientation, and the orientation of the lens surface. The $h_{i,j}$ in~\eqref{equ:e10} can be simplified as 
    \begin{align}\label{channel_sim}
        h_{i,j} = k_1||\hat{\boldsymbol{P}}_{T,i}-\hat{\boldsymbol{P}}_{R,j}||^{-(m+2)} \text{c}\left(\phi_{i,j}\right),
    \end{align}
where the term $||\hat{\boldsymbol{P}}_{T,i}-\hat{\boldsymbol{P}}_{R,j}||^{-(m+2)}$ remains invariant for a given receiver position, and $\phi_{i,j}$ represents the light incident angle at the PD after refraction from the liquid surface. The angle $\phi_{i,j}$ is the angle between the unit vectors $-\hat{\boldsymbol{\eta}}_{ref,i,j}$ and $\hat{\boldsymbol{\eta}}_{R,j}$. Consequently, $\text{c}(\phi_{i,j})$ can be expressed as
    \begin{align}
        \text{c}(\phi_{i,j}) = -\hat{\boldsymbol{\eta}}_{ref,i,j}\cdot\hat{\boldsymbol{\eta}}_{R,j}.
    \end{align}
By taking the dot product of \eqref{equ:ref_o} with $-\hat{\boldsymbol{\eta}}_{R,j}$ yields
    \begin{align}\label{cos_theta}
        \text{c}(\phi_{i,j}) \hspace{-0.5mm}= \hspace{-0.5mm}-{n_{l}}^{-1}\big(\big(\text{c}(\alpha_{i,j})-(n_l^2 \hspace{-0.5mm}&- \hspace{-0.5mm}\text{s}^2(\alpha_{i,j}))^{\frac{1}{2}}\big)(\hat{\boldsymbol{\eta}}_{len,j}\cdot \hat{\boldsymbol{\eta}}_{R,j}\big)\hspace{-0.5mm}\nonumber \\
        &-\hspace{-0.5mm}\big(\hat{\boldsymbol{\eta}}_{TR,i,j}\cdot\hat{\boldsymbol{\eta}}_{R,j}\big)\big).
    \end{align}
We observe that adjusting the values of $\psi_x$ and $\psi_y$ to maximize $\text{c}(\phi_{i,j})$ in \eqref{cos_theta} can help mitigate the performance degradation caused by random receiver orientation of the receiver. 

\section{TLL Orientation Schemes}\label{lens_orientation}
In this section, we propose three TLL orientation schemes aimed at enhancing signal reception performance in the proposed TLL-assisted VLC system (Fig. \ref{fig:TLL_schemes}). Specifically, the TLL orientation angles, $\psi_{x,j}$ and $\psi_{y,j}$, are dynamically adjusted according to the receiver position,  the receiver orientation, and the position of the \textit{tagged AP}. The BSR scheme maximizes the received signal strength from the \textit{tagged AP} at a user device while suppressing interference from neighboring APs. The CLS and VULO schemes offer lower computational complexity and still achieve notable performance improvements over conventional fixed-orientation receivers. We analyze the behavior of $\psi_{x,j}$ and $\psi_{y,j}$ under each scheme, as well as the corresponding $\text{c}(\phi_{i,j})$ values, which will be instrumental in the outage probability analysis presented in Section~\ref{outage}.

\subsection{BSR Scheme}
In the BSR scheme, we adjust the TLL orientation to maximize the channel gain of the VLC link from the \textit{tagged AP} to the receiver. It is observed that the channel gain, $h_{i,j}$, reaches its maximum value when $\text{c}(\phi_{i,j}) = 1$ in \eqref{channel_sim} for fixed transmitter and receiver positions (Fig. 3(a)). To determine the corresponding TLL orientation angles, $\psi_{x,j}$ and $\psi_{y,j}$, the following steps are used. First, we set $\text{c}(\phi_{i,j}) = 1$ in \eqref{cos_theta}. By expanding $\hat{\boldsymbol{\eta}}_{L,j}\hspace{-1.5mm}\cdot\hat{\boldsymbol{\eta}}_{R,j}$ and using the fundamental trigonometric identity, $\text{c}^2{\theta}+\text{s}^2{\theta} = 1$, we can express $\hat{\boldsymbol{\eta}}_{L,j}\hspace{-1.5mm}\cdot\hat{\boldsymbol{\eta}}_{R,j}$ as $\text{c}(\psi_{x,j})\text{c}(\psi_{y,j})$. Consequently, in the BSR scheme, the expression in \eqref{cos_theta} can be simplified to 
\begin{align}\label{cos_theta_1}
        \hspace{-2mm}n_l\hspace{-0.8mm}- \hspace{-0.8mm}\hat{\boldsymbol{\eta}}_{RT,i,j}\hspace{-0.9mm}\cdot \hspace{-0.7mm}\hat{\boldsymbol{\eta}}_{R,j}\hspace{-0.8mm}+\hspace{-0.8mm} \Big(\hspace{-0.8mm}\text{c}\alpha_{i,j}\hspace{-0.9mm}-\hspace{-0.9mm}(n_l^2 \hspace{-0.9mm}- \hspace{-0.9mm} \text{s}^2\alpha_{i,j})^{\frac{1}{2}}\hspace{-0.8mm}\Big)\text{c}\psi_{x,j}\text{c}\psi_{y,j}\hspace{-0.8mm}=\hspace{-0.8mm}0.
    \end{align}
We use numerical methods to solve \eqref{cos_theta_1} and obtain candidate optimal values of $\psi_{x,j}$ and $\psi_{y,j}$, respectively. A feasible value for $\psi_{x,j}\in[\psi_x^{L},\psi_x^{H}]$ is selected, and~\eqref{cos_theta_1} is solved to obtain $\psi_{y,j}\in[\psi_y^{L},\psi_y^{H}]$. This is repeated until we obtain a feasible solution. Here $\psi_x^{L}$, $\psi_x^{H}$, $\psi_y^{L}$, and $\psi_y^{H}$ are the limits of the TLL orientation angles imposed by the hardware limitations. Even though an iterative search is required to find the optimal TLL orientation angles, BSR scheme requires only knowledge of the relative transmitter/receiver positions as well as the receiver orientation angles as inputs.

\begin{figure}
    \centering
    \includegraphics[width=0.95\linewidth]{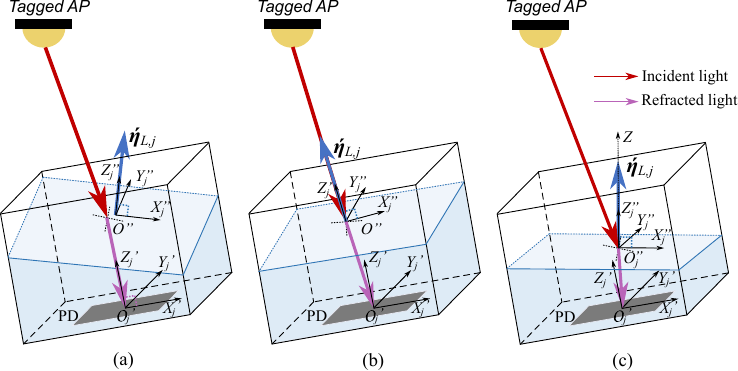}
    \caption{TLL orientation schemes. (a) BSR. (b) CLS. (c) VULO.}
    \vspace{-2mm}
    \label{fig:TLL_schemes}
\end{figure}

\subsection{CLS Scheme}
In the CLS scheme, we select the TLL orientation such that the lens is always pointed towards the \textit{tagged AP}, regardless of the receiver position and the receiver orientation. The motivation for this low-complexity scheme is that it simplifies the alignment process without an iterative search, as in the BSR scheme. In this scheme, $\hat{\boldsymbol{\eta}}_{L,j} = 
\hat{\boldsymbol{\eta}}_{RT,i,j}$ and $\alpha_{i,j} = 0^{\circ}$ (Fig. 3(b)). Using \eqref{cos_theta}, $\text{c}\left(\phi_{i,j}\right)$ for this scheme is given by
\begin{align}\label{cos_theta_2}
        \text{c}(\phi_{i,j}) = \hat{\boldsymbol{\eta}}_{RT,i,j}\cdot\hat{\boldsymbol{\eta}}_{R,j}.
    \end{align}
To find the optimal $\psi_{x,j}$ and $\psi_{y,j}$, the following steps can be used. In this scheme, the directional vectors of $\hat{\boldsymbol{\eta}}_{L,j}$ and $\hat{\boldsymbol{\eta}}_{RT,i,j}$ have equal components along the $x$, $y$, and $z$ directions. Components along the $x$ and $y$ directions can be combined to obtain the following expression
    \begin{align}\label{exp111}
        \text{s}\phi_{R,j}\text{c}\psi_{y,j}\text{c}\psi_{x,j}\hspace{-0.6mm}-\hspace{-0.6mm}\text{c}\phi_{R,j}\text{s}\psi_{y,j} \hspace{-0.7mm}=\hspace{-0.7mm} e_{RT,x,i,j}\text{c}\theta_{R,j}\hspace{-0.6mm}+\hspace{-0.6mm}e_{RT,y,i,j}\text{s}\theta_{R,j}.
    \end{align}
By combining \eqref{exp111} with the component along the $z$ direction, we can find $\psi_{y,j}$ as
    \begin{align}\label{exp112}
        \psi_{y,j} = \sin^{-1}\left(\hat{\boldsymbol{\eta}}_{RT,i,j}\cdot\hat{\boldsymbol{\eta}}_{R,t,j}\right),
    \end{align}
where $\hat{\boldsymbol{\eta}}_{R,t,j}=[-\text{c}{\theta_{R,j}}\text{c}{\phi_{R,j}},-\text{s}{\theta_{R,j}}\text{c}{\phi_{R,j}},\text{s}{\phi_{R,j}}]^T$. Next, to obtain the optimal $\psi_{x,j}$, the components along the $x$ and $y$ directions are combined as
    \begin{align}\label{exp113}
        -\text{c}(\psi_{y,j})\text{s}(\psi_{x,j})= e_{RT,x,i,j}\text{s}(\theta_{R,j})-e_{RT,y,i,j}\text{c}(\theta_{R,j}).
    \end{align}
Substituting \eqref{exp112} into \eqref{exp113}, $\psi_{x,j}$ can be expressed as
    \begin{align}\label{exp114}
        \psi_{x,j} = \sin^{-1}\Bigg(\frac{e_{RT,y,i,j}\text{c}(\theta_{R,j})-e_{RT,x,i,j}\text{s}(\theta_{R,j})}{\sqrt{1-\left(\hat{\boldsymbol{\eta}}_{RT,i,j}\cdot         \hat{\boldsymbol{\eta}}_{R,t,j}\right)^2}}\Bigg).
    \end{align}
Hence, the optimal TLL orientation angles can be obtained using \eqref{exp112} and \eqref{exp114} in this scheme. Similar to the BSR scheme, the CLS scheme requires only knowledge of the relative transmitter/receiver positions and the receiver orientation angles as inputs. However, an iterative search is not necessary, making CLS a low-complexity lens orientation scheme.

\subsection{VULO Scheme}
In the VULO scheme, the orientation of the lens surface is kept vertically upward, regardless of the receiver orientation and position (Fig. 3(c)). This approach is considered as a low-complexity scheme compared to the BSR and CLS schemes. In particular, no receiver position or orientation information are required as in BSR and CLS schemes. In this scheme, $\hat{\boldsymbol{\eta}}_{L,j} = 
[0, 0, 1]^T$, which results in $\text{c}(\alpha_{i,j}) = e_{RT,z,i,j}$. Additionally, $\hat{\boldsymbol{\eta}}_{L,j}\cdot\hat{\boldsymbol{\eta}}_{R,j} = \text{c}(\phi_{R,j})$. As a result, $\text{c}(\phi_{i,j})$ is simplified to
    \begin{align}\label{cos_theta_3}
        \text{c}(\phi_{i,j}) \hspace{-0.5mm}=\hspace{-0.5mm} -{n_{l}^{-1}}\hspace{-0.5mm}\big(e_{RT,z,i,j}\hspace{-0.5mm}&-\hspace{-0.5mm}(n_l^2 \hspace{-0.5mm}-\hspace{-0.5mm} (1-e_{RT,z,i,j}^2)))^{\frac{1}{2}}\hspace{-0.5mm}\big)\text{c}(\phi_{R,j})\nonumber \\
        &+n_{l}^{-1}\left(\hat{\boldsymbol{\eta}}_{RT,i,j}\cdot \hat{\boldsymbol{\eta}}_{R,j}\right).
    \end{align}
To determine the optimal values of $\psi_{x,j}$ and $\psi_{y,j}$, the following steps are employed. In this scheme, the components of $\hat{\boldsymbol{\eta}}_{L,j}$ along the $x$- and $y$-axes are set to zero. By multiplying the $x$-axis component by $\text{c}(\theta_{R,j})$ and the $y$-axis component by $\text{s}(\theta_{R,j})$ in \eqref{fun:len}, and subsequently summing the resulting expressions, we derive
    \begin{align}\label{equ_122}
        \text{s}(\phi_{R,j})\text{c}(\psi_{y,j})\text{c}(\psi_{x,j})-\text{c}(\phi_{R,j})\text{s}(\psi_{y,j})=0.
    \end{align}
Next, we combine \eqref{equ_122} with the component along the $z$-axis of $\hat{\boldsymbol{\eta}}_{L,j}$ in \eqref{fun:len} to solve for the optimal $\psi_{y,j}$. In the VULO scheme, the optimal $\psi_{y,j}$ is given by $\phi_{R,j}$. By substituting $\psi_{y,j} = \phi_{R,j}$ into \eqref{equ_122}, we obtain the optimal $\psi_{x,j}=0^{\circ}$.

\section{Outage Probability Analysis}\label{outage}
In this section, we derive the outage probability of the proposed system and provide closed-form or approximate results for each lens orientation scheme. The outage probability of a user located at position $\mathbf{y}$ is defined as the probability that the received SINR falls below a target threshold $\gamma$ \textit{i.e.,}
\begin{align}
    P_o(\gamma,\mathbf{y}) = \mathbb{P}[\text{SINR}_{0}(\mathbf{y})<\gamma].
\end{align}
The following proposition provides a general outage expression valid for all lens-orientation schemes.

\begin{proposition}\label{Proposition1}
For $\gamma\geq1$, the conditional outage probability of a user located at $\mathbf{y}$ can be expressed as
\begin{align}\label{eq_Pout}
    P_o(\gamma,\mathbf{y}) &= \frac{2-\lambda_m|\mathcal{A}_{\mathbf{y}}|}{2}+ \frac{\lambda_m}{\pi} \int_0^{\infty}\frac{1}{t}\nonumber\\
    &\times\text{Im}\Big[\exp(jt\sigma'^2)\mathcal{F}(-jt,\mathcal{A}_{\mathbf{y}})\mathcal{K}(-jt,\bar{\mathcal{A}}_{\mathbf{y}})\Big]dt,
\end{align}
where
\begin{align} \label{F_SA}
 \mathcal{F}(\hspace{-0.5mm}-\hspace{-0.5mm}jt,\mathcal{A}_{\mathbf{y}}\hspace{-0.5mm}) \hspace{-0.6mm}= \hspace{-1.4mm}\int_{\mathcal{A}_{\mathbf{y}}} \hspace{-3mm}\exp\left(\hspace{-0.5mm}-jt\gamma^{-1}\hspace{-0.5mm}\text{c}^2{\phi_{0,j}}(||\mathbf{x}||^2+h_u^2)^{-(m+2)}\right)d\mathbf{x},
\end{align}
and
\begin{align} \label{K_SA}
\hspace{-2mm}\mathcal{K}(\hspace{-0.7mm}-jt,\bar{\mathcal{A}}_{\mathbf{y}}\hspace{-0.7mm}) \hspace{-1mm}=\hspace{-1mm} \exp\hspace{-0.5mm}\bigg(\hspace{-0.5mm}-\lambda_m\hspace{-1mm}\int_{\bar{\mathcal{A}}_{\mathbf{y}}}\hspace{-2.5mm} \hspace{-0.5mm}1\hspace{-0.5mm}-\hspace{-0.5mm}\exp\bigg(\hspace{-0.5mm}\frac{-jt\text{c}^2(\phi_{\mathbf{z},0})}{(||\mathbf{z}||^2\hspace{-0.7mm}+\hspace{-0.7mm}h_u^2)^{m\hspace{-0.5mm}+\hspace{-0.5mm}2}}\hspace{-0.5mm}\bigg)\hspace{-0.2mm}d\mathbf{z}\hspace{-0.5mm}\bigg).
\end{align}
\end{proposition}

\begin{proof}
See Appendix \ref{Appendix1}.
\end{proof}

Finally, the outage probability for a user with arbitrary location and receiver orientation is
\begin{align}\label{P_outf1}
    \hspace{-0.9mm}P_o(\gamma) \hspace{-0.8mm}=\hspace{-0.8mm} \frac{1}{2\pi X_mY_m}\hspace{-1.2mm}\int_{0}^{\pi/2}\hspace{-6mm}f_{\phi_R}(\phi_R)\hspace{-0.9mm}\int_{0}^{2\pi}\hspace{-3mm}\int_{\mathcal{S}(0,X_m,Y_m)}\hspace{-14mm}P_o(\gamma,\mathbf{y})d\mathbf{y} d\theta_R d\phi_R,
\end{align}
where $f_{\phi_R}(\phi_R)$ denotes the probability density function (PDF) of the random variable (RV) $\phi_R$.

\subsection{Outage Probability of BSR Scheme}
Here, we derive approximate closed-form expressions for the outage probability of the BSR scheme. In particular, when the signal from the selected AP impinges perpendicularly on the PD, we have $\text{c}(\phi_{i,0})=1$. To this end, we first simplify \eqref{F_SA} and \eqref{K_SA}, and then use \eqref{eq_Pout} and \eqref{P_outf1} to obtain the final outage probability expressions. 

\begin{proposition}\label{Proposition2}
For the BSR scheme, \eqref{F_SA} can be simplified as
\begin{align}
\hspace{-2.8mm}\mathcal{F}(\hspace{-0.5mm}-jt,\mathcal{A}_{\mathbf{y}}\hspace{-0.5mm}) \hspace{-0.5mm}= \hspace{-0.5mm}\hspace{-0.5mm}\pi\hspace{-0.9mm} 
\sum_{k=0}^{\infty}\hspace{-0.9mm} \frac{(\hspace{-0.5mm}-\alpha_t\hspace{-0.5mm})^kh_u^{2(1+k/k_m)}}{k!(1+k/k_m)}\hspace{-0.5mm}\bigg[\hspace{-0.5mm}\Big(\hspace{-0.5mm}\frac{a_D^2}{h_u^2}\hspace{-0.8mm}+\hspace{-0.5mm}1\Big)^{1+k/k_m}\hspace{-4.5mm}-\hspace{-0.9mm}1\bigg],\label{F_final_BSR}
\end{align}
where $\alpha_t = jt\gamma^{-1}$ and $k_m = -1/(m+2)$.
\end{proposition}

\begin{proof}
See Appendix \ref{Appendix2}.
\end{proof}

\begin{proposition}\label{Proposition3}
For the BSR scheme, the interference term \eqref{K_SA} under the weak interferer-alignment assumption can be approximated as
\begin{equation}
\mathcal K(-jt,\bar{\mathcal A}_{\mathbf y}) \approx
\exp(-jt\lambda_m\Xi_{\rm BSR}),\label{K_closed_BSR}
\end{equation}
where
\begin{equation}
\hspace{-2mm}\Xi_{\rm BSR}
 =
  \frac{
  \pi\bar G_{\rm a}
  }{
  m+1
  }
  \left[
  (R^2\hspace{-1mm}+\hspace{-1mm}h_u^2)^{-(m+1)}
  \hspace{-1mm}-\hspace{-1mm}
  (a_D^2\hspace{-1mm}+\hspace{-1mm}h_u^2)^{-(m+1)}
  \right],
  \label{Xi_BSR}
\end{equation}
and
        \begin{equation}
        \hspace{-1.4mm}\bar G_{\mathrm{a}}
        \hspace{-0.5mm}\approx\hspace{-0.5mm}
        \frac1{n_l^2}\hspace{-0.5mm}
        \bigg[\hspace{-0.5mm}
        \frac12
       \hspace{-0.9mm} +\hspace{-0.9mm}
        \left(\hspace{-0.5mm}
        n_l^2\hspace{-0.7mm}-\hspace{-0.7mm}\frac23\hspace{-0.5mm}
        \right)\hspace{-0.9mm}
        \frac{
        1\hspace{-0.5mm}+\hspace{-0.5mm}
        \exp(-2\sigma_{\phi_R}^2)
        \text{c}(2\mu_{\phi_R})
        }{2}
        \bigg]\hspace{-1mm}.
        \label{G_avg}
        \end{equation}
\end{proposition}

\begin{proof}
See Appendix \ref{Appendix3}.
\end{proof}

Using Propositions 2 and 3, we can state the following theorem.

\begin{theorem}\label{Theorem1}
The outage probability for the BSR scheme under weak interferer-alignment can be approximated by \eqref{P_BSR_final}.
\end{theorem}

\begin{proof}
See Appendix \ref{Appendix4}.
\end{proof}

    \begin{figure*}[!t]
    \begin{align}
    P_o^{\rm BSR}(\gamma)
    &\approx
    1
    -
    \frac{\lambda_m\pi}{2}
    \left[
    (\gamma\sigma'^2)^{-\frac1{m+2}}
    -h_u^2
    \right]+
    \lambda_m
    \sum_{k=1}^{K_{\max}}
    \frac{
    \gamma^{-k}
    \Gamma(k)
    \text{s}\!\left(\frac{k\pi}{2}\right)
    }{
    k!\left(1-k(m+2)\right)
    }
    \bigg(
    \sigma'^2
    -
    \frac{\lambda_m\pi}{(m+1)n_l^2}
    \left[
    (R^2+h_u^2)^{-(m+1)}
    -
    (\gamma\sigma'^2)^{\frac{m+1}{m+2}}
    \right]\nonumber\\
    &\times
    \bigg[
    \frac12
    +
    \left(
    n_l^2-\frac23
    \right)
    \frac{
    1+\exp(-2\sigma_{\phi_R}^2)\text{c}(2\mu_{\phi_R})
    }{2}
    \bigg]
    \bigg)^{-k}
    \left[
    (\gamma\sigma'^2)^{\frac{k(m+2)-1}{m+2}}
    -
    h_u^{2(1-k(m+2))}
    \right].
    \label{P_BSR_final}
    \end{align}\vspace{-1mm}
    \hrule
     \vspace{-3mm}
    \end{figure*}
    
\subsection{Outage Probability of CLS Scheme}
In this subsection, we derive approximate closed-form expressions for the outage probability of the CLS scheme. When the lens is oriented toward the closest AP, we have $\text{c}\phi_{i,0} = \hat{\boldsymbol{\eta}}_{RT,i,0}\cdot \hat{\boldsymbol{\eta}}_{R,0}$. To this end, we first simplify the expressions for \eqref{F_SA} and \eqref{K_SA}, and subsequently employ \eqref{eq_Pout} and \eqref{P_outf1} to obtain the final outage probability expressions.

\begin{proposition}\label{Proposition4}
    Under the CLS scheme and the practical assumption $a_D \le \min\left\{\frac{X_m}{2},\,\frac{Y_m}{2}\right\}$, the function $\mathcal{F}(-jt,\mathcal{A}_{\mathbf{y}})$ can be approximated as
        \begin{align}
        &\mathcal{F}(-jt,\mathcal{A}_{\mathbf{y}}) \hspace{-0.5mm}\approx\hspace{-0.5mm} \frac{\pi}{m+3} 
        \hspace{-0.5mm}\sum_{k=0}^{\infty} 
        \frac{(jt\gamma^{-1}\text{s}\phi_RC_c)^{2k}}{(k!)^2} 
        (D_c)^{\frac{1}{m+3}-2k}\nonumber \\
        &\times\Bigl[
        \Gamma\Bigl(2k - \frac{1}{m+3} ,\; D_c\left(a_D^2 + h_u^2\right)^{-(m+3)}\Bigr)
        \nonumber \\
        &-
        \Gamma\Bigl(2k - \frac{1}{m+3} ,\; D_c h_u^{-2(m+3)}\Bigr)
        \Bigr], 
        \label{equ:F_CLS1}
    \end{align}
    where $D_c = jt\gamma^{-1}(C_c^2+a_D^2\text{s}^2\phi_R/4)$ and $C_c = h_u\text{c}\phi_R$.
    \end{proposition}

\begin{proof}
See Appendix \ref{Appendix5}.
\end{proof}

\begin{proposition}\label{Proposition5}
    Under the PPP approximation of the MHCPP, the interference term $\mathcal{K}(-jt,\bar{\mathcal{A}}_{\mathbf{y}})$ can be approximated as
        \begin{align}\label{equ:K_CLS1}
            &\mathcal{K}(-jt,\bar{\mathcal{A}}_{\mathbf{y}})\nonumber\\&\approx
            \exp\hspace{-0.7mm}\Big(\hspace{-1.5mm}
            -\hspace{-1.0mm}\frac{\lambda_m jt\,\pi}{2 n_l^2}
            \hspace{-1.5mm}\sum_{k=m+1}^{m+2}\hspace{-3mm}
            D_k\hspace{-0.7mm}
            \left[
            (R^2\hspace{-1mm}+\hspace{-1mm}h_u^2)^{-k}
            \hspace{-1mm}-\hspace{-1mm}
            (a_D^2\hspace{-1mm}+\hspace{-1mm}h_u^2)^{-k}
            \right]
            \Big),
        \end{align}
    where $\bar{D}_{m+1} \hspace{-0.5mm}=\hspace{-0.5mm} (n_l^2\hspace{-0.9mm}-\hspace{-0.9mm}1)\hspace{-0.9mm}\left[1\hspace{-0.5mm}-\hspace{-0.5mm}\pi\lambda_m h_u^2\exp(\pi\lambda_m h_u^2)E_1(\pi\lambda_m h_u^2)\right]\hspace{-0.9mm}+\hspace{-0.9mm}\text{s}^2\phi_R$, and $\bar{D}_{m\hspace{-0.5mm}+\hspace{-0.5mm}2} \hspace{-0.9mm}=\hspace{-0.9mm} 2h_u^2\hspace{-0.9mm}\big(\hspace{-0.5mm}(n_l^2\hspace{-1.2mm}-\hspace{-1.2mm}1\hspace{-0.5mm})\pi\lambda_m h_u^2\exp(\hspace{-0.5mm}\pi\lambda_m h_u^2)E_1(\pi\lambda_m h_u^2)\hspace{-0.5mm}$ $+\hspace{-0.5mm} \text{c}^2\phi_R\big)$.
\end{proposition}

\begin{proof}
See Appendix \ref{Appendix6}.
\end{proof}

The outage probability for the CLS scheme is given in the following theorem.

\begin{theorem}\label{Theorem2}
The outage probability for the CLS scheme under PPP approximation can be approximated by \eqref{eq:Pout_CLS_final}, where $\bar{\beta} = [(h_u^2+a_D^2/4)+ (h_u^2-a_D^2/4)\exp(-2\sigma^2_{\phi_R})\text{c}(2\mu_{\phi_R})]$, and $\bar{\Xi}_{\mathrm{CLS}}
    \hspace{-0.5mm}=\hspace{-0.5mm}
    \frac{\lambda_m\pi}{2n_l^2}
    \sum_{q=m+1}^{m+2}
    \bar D_q
    \left[
    (R^2+h_u^2)^{-q}
    -
    (a_D^2+h_u^2)^{-q}
    \right]$.
\end{theorem}

\begin{proof}
    See Appendix \ref{Appendix7}.
\end{proof}

\begin{figure*}[!t]
    \begin{align}
    P_o^{\mathrm{CLS}}(\gamma)
    \hspace{-0.5mm}\approx\hspace{-0.5mm}
    1\hspace{-0.7mm}
    -\hspace{-0.7mm}
    \frac{\lambda_m\pi}{2}\hspace{-0.8mm}
    \left[
    (\gamma\sigma'^2)^{-\frac1{m+2}}
    \hspace{-0.5mm}-\hspace{-0.5mm}h_u^2
    \right]\hspace{-0.8mm}+\hspace{-0.8mm}
    \frac{\lambda_m}{m+3}\hspace{-1.5mm}
    \sum_{k=0}^{K_{\max}}\hspace{-0.5mm}
    \frac{
    \left(\hspace{-0.5mm}
    \frac{h_u^2}{8}\hspace{-0.5mm}
    \left[
    1\hspace{-0.8mm}-\hspace{-0.8mm}\exp(-8\sigma_{\phi_R}^2)
    \text{c}(4\mu_{\phi_R})
    \right]\hspace{-0.5mm}
    \right)^k
    \bar\beta^{\frac1{m+3}-2k}
    }{
    (k!)^2
    }
    \frac{
    (\bar\beta_2^\mu\hspace{-0.5mm}-\hspace{-0.5mm}\bar\beta_1^\mu)
    \Gamma(2\mu)
    }{
    \mu
    (\sigma_{\mathrm{eff}}^2)^{2\mu}
    }
    \text{s}\hspace{-0.5mm}
    \left(\hspace{-0.5mm}
    \theta_k\hspace{-1mm}+\hspace{-1mm}\frac{3\mu\pi}{2}\hspace{-0.5mm}
    \right)\hspace{-1.2mm},
    \label{eq:Pout_CLS_final}
    \end{align}
    
    \vspace{-8mm}
            \begin{align}\label{final_out_VULO}
        \hspace{-3mm}P_o^{\mathrm{VULO}}(\gamma)
        &\hspace{-0.9mm}\approx\hspace{-0.9mm}
        1
        \hspace{-0.9mm}-\hspace{-0.9mm}
        \frac{\lambda_m \pi}{2}\hspace{-0.9mm}
        \left[
        (\gamma\sigma'^2)^{-\frac{1}{m+2}}
        \hspace{-0.85mm}-\hspace{-0.8mm}h_u^2
        \right]\hspace{-0.9mm}+\hspace{-0.9mm}
        \frac{\lambda_m}{\pi}\hspace{-0.5mm}
        \sum_{k=0}^{K_{\max}}
        \frac{(-1)^k}{\epsilon^{2k+1}}\hspace{-0.5mm}
        \Bigg[\hspace{-0.5mm}
        \frac{\hspace{-0.5mm}
        \pi\hspace{-0.5mm}
        \left[
        (\gamma\sigma'^2)^{-\frac{1}{m+2}}
        \hspace{-0.5mm}-\hspace{-0.5mm}h_u^2
        \right]\hspace{-0.9mm}
        \left(
        \sigma'^2\hspace{-0.5mm}-\hspace{-0.5mm}\bar{\Xi}_{\mathrm{VULO}}
        \right)^{\hspace{-0.7mm}2k+1}
        }{
        2k+1
        }
        \hspace{-0.9mm}+\hspace{-0.5mm}
        \bar{\Xi}_F
        \hspace{-0.9mm}\left(
        \sigma'^2\hspace{-0.9mm}-\hspace{-0.9mm}\bar{\Xi}_{\mathrm{VULO}}
        \right)^{\hspace{-0.9mm}2k}\hspace{-0.9mm}
        \Bigg]\hspace{-0.5mm}.
        \end{align}
                \vspace{-1mm}
            \hrule
        \vspace{-4mm}
\end{figure*}

\subsection{Outage Probability of VULO Scheme}
In the VULO case, the optical lens is oriented vertically upward, which simplifies the geometrical relationship between the receiver orientation and the incoming optical signal. Using this orientation model, analytical expressions for $\mathcal{F}(-jt,\mathcal{A}_\mathbf{y})$ and $\mathcal{K}(-jt,\bar{\mathcal{A}}_\mathbf{y})$ are obtained to derive the outage probability. The simplified form of $\text{c}\phi_{i,0}$ is obtained by substituting~\eqref{equ:rot011} into~\eqref{cos_theta_3}. After straightforward mathematical manipulations, $\text{c}\phi_{i,0}$ can be expressed as
\begin{align}\label{eq13}
    \text{c}\phi_{i,0} &= n_{l}^{-1}\Big(\text{c}\phi_{R}\sqrt{n_l^2 - (1-{e_{RT,z,i,0}}^2)}\nonumber \\
    &+e_{RT,x,i,0}\text{c}\theta_{R} \text{s}\phi_{R}+ e_{RT,y,i,0}\text{s}\theta_{R} \text{s}\phi_{R}\Big).
\end{align}

\begin{proposition}\label{Proposition6}
        The term $\mathcal{F}(\hspace{-0.5mm}-jt,\mathcal{A}_\mathbf{y}\hspace{-0.5mm})$ for the VULO scheme can be approximated as
        \begin{align}\label{F_VULO_final}
        &\mathcal{F}(\hspace{-0.5mm}-jt,\mathcal{A}_\mathbf{y}\hspace{-0.5mm})
        \hspace{-0.5mm}\approx\hspace{-0.5mm}
        \pi a_D^2
        \hspace{-0.9mm}+\hspace{-0.7mm}\frac{jt\pi h_u \text{c}^2\theta_R}{\gamma(m+2)}\hspace{-0.7mm}
        \Big(\hspace{-0.7mm}
        h_u^{-2(m+2)}
        \hspace{-1.4mm}-\hspace{-0.7mm}(a_D^2\hspace{-0.9mm}+\hspace{-0.9mm}h_u^2)^{-(m+2)}\hspace{-0.5mm}
        \Big)
        \nonumber\\
        &
        -\frac{jt\pi  (n_l^2\hspace{-0.5mm}-\hspace{-0.5mm}1)\text{c}^2\theta_R}{\gamma n_l^2}\hspace{-0.5mm}
        \bigg[\hspace{-0.5mm}
        \frac{1}{m\hspace{-0.5mm}+\hspace{-0.5mm}1}
        \Big(
        h_u^{-2(m+2)}
        \hspace{-0.5mm}-\hspace{-0.5mm}(a_D^2\hspace{-0.5mm}+\hspace{-0.5mm}h_u^2)^{-(m+1)}\hspace{-0.5mm}
        \Big)
        \nonumber\\
        &
        -\frac{h_u^2}{m+2}
        \Big(
        h_u^{-2(m+2)}
        -(a_D^2+h_u^2)^{-(m+2)}
        \Big)
        \bigg].
        \end{align}
    \end{proposition}

\begin{proof}
    See Appendix \ref{Appendix8}.
\end{proof}

    \begin{proposition}\label{Proposition7}
        Under the VULO scheme, $\mathcal{K}(-jt,\bar{\mathcal{A}}_{\mathbf{y}})$ can be approximated by
        \begin{align}\label{K_x}
            &\mathcal{K}(-jt,\bar{\mathcal{A}}_{\mathbf{y}})
            \hspace{-0.7mm}=\hspace{-0.7mm}
            \exp\hspace{-0.5mm}\bigg(\hspace{-1.8mm}
            -\hspace{-0.5mm}\frac{\pi jt\lambda_m\text{s}^2\phi_R}{\gamma n_l}\hspace{-0.7mm}
            \bigg[\hspace{-0.7mm}
            \frac{(R^2\hspace{-0.7mm}+\hspace{-0.7mm}h_u^2)^{-m}\hspace{-0.7mm}-\hspace{-0.7mm}(a_D^2\hspace{-0.7mm}+\hspace{-0.7mm}h_u^2)^{-m}}{2m}
            \nonumber\\
            &
            -\frac{h_u^2}{2(m+1)}
            \big(
            (R^2+h_u^2)^{-(m+1)}
            -(a_D^2+h_u^2)^{-(m+1)}
            \big)
            \bigg]
            \nonumber\\
            &
            -\frac{2\pi jt\lambda_m\text{c}^2\phi_R}{\gamma n_l}
            h_u^{-2(m+1)}\hspace{-0.5mm}
            \sum_{l=0}^{\infty}\hspace{-0.5mm}
            \bigg[\hspace{-0.5mm}
            \sum_{k=0}^{l}
            (-1)^l
            \frac{1}{(2l\hspace{-0.5mm}+\hspace{-0.5mm}2)h_u^{2l}}
            \binom{m\hspace{-0.5mm}+\hspace{-0.5mm}2}{k}
            \nonumber\\
            &\times
            \binom{1}{l-k}
            \left(
            \frac{n_l^2-1}{n_l^2}
            \right)^{l-k}
            \bigg]
            \left(
            a_D^{2l+2}-R^{2l+2}
            \right)\bigg).
        \end{align}
    \end{proposition}

\begin{proof}
    See Appendix \ref{Appendix9}.
\end{proof}

We can now state the following theorem.

\begin{theorem}\label{Theorem3}
The outage probability in VULO scheme is given by \eqref{final_out_VULO}, where $\bar{\Xi}_{\text{VULO}}$ and $\bar{\Xi}_{\text{F}}$ are given by
    \begin{align}
        &\bar{\Xi}_{\mathrm{VULO}}
        \hspace{-0.5mm}=\hspace{-0.5mm}
        \frac{\pi\lambda_m\bar{s}_{\phi_R}^{\,2}}{\gamma n_l}\hspace{-0.5mm}
        \bigg[\hspace{-0.5mm}
        \frac{\hspace{-0.5mm}
        (R^2\hspace{-0.5mm}+\hspace{-0.5mm}h_u^2)^{-m}
        \hspace{-0.5mm}-\hspace{-0.5mm}
        (a_D^2\hspace{-0.5mm}+\hspace{-0.5mm}h_u^2)^{-m}
        }{2m}  \hspace{-0.5mm}-\hspace{-0.5mm}
        \frac{h_u^2}{2(m\hspace{-0.5mm}+\hspace{-0.5mm}1)}
        \nonumber\\
        &\hspace{-1mm}\times\hspace{-1mm}
        \Big(\hspace{-0.9mm}
        (R^2\hspace{-0.9mm}+\hspace{-0.9mm}h_u^2)^{-(m+1)}
        \hspace{-0.8mm}-\hspace{-0.8mm}
        (a_D^2\hspace{-0.9mm}+\hspace{-0.9mm}h_u^2)^{-(m+1)}
        \Big)\hspace{-0.9mm}
        \bigg]\hspace{-1mm}+\hspace{-0.9mm}
        \frac{2\pi\lambda_m\bar{c}_{\phi_R}^{\,2}}{\gamma n_l}
        h_u^{-2(m+1)}
        \nonumber\\
        &\sum_{l=0}^{\infty}\hspace{-0.9mm}
        \bigg[\hspace{-0.9mm}
        \sum_{q=0}^{l}\hspace{-0.9mm}
        \frac{(-1)^l}{(2l\hspace{-0.5mm}+\hspace{-0.5mm}2)h_u^{2l}}
        \binom{m\hspace{-0.9mm}+\hspace{-0.9mm}2}{q}\hspace{-1.3mm}
        \binom{1}{l\hspace{-0.5mm}-\hspace{-0.5mm}q}\hspace{-1.3mm}
        \left(\hspace{-0.9mm}
        \frac{n_l^2\hspace{-0.9mm}-\hspace{-0.9mm}1}{n_l^2}
        \right)^{\hspace{-1.3mm}l-q}\hspace{-0.9mm}
        \bigg]\hspace{-0.9mm}
        \left(
        a_D^{2l+2}\hspace{-1.3mm}-\hspace{-1.3mm}\hspace{-0.9mm}R^{2l+2}
        \right)\hspace{-0.9mm},
        \end{align}
    and 
                \begin{align}
            \bar{\Xi}_F
            &=
            \frac{\pi h_u \bar{c}_{\phi_R}^{\,2}}{m+2}
            \left(
            h_u^{-2(m+2)}
            -
            (a_D^2+h_u^2)^{-(m+2)}
            \right)
            \nonumber\\
            &-
            \frac{\pi(n_l^2-1)\bar{c}_{\phi_R}^{\,2}}{n_l^2}
            \bigg[
            \frac{
            h_u^{-2(m+1)}
            -
            (a_D^2+h_u^2)^{-(m+1)}
            }{m+1}
            \nonumber\\
            &
            -
            \frac{h_u^2}{m+2}
            \left(
            h_u^{-2(m+2)}
            -
            (a_D^2+h_u^2)^{-(m+2)}
            \right)
            \bigg].
            \end{align}
\end{theorem}

\begin{proof}
See Appendix \ref{Appendix10}.
\end{proof}
    
\section{Numerical Results and Discussion}\label{results}
In this section, numerical results are presented to validate the accuracy of the derived analytical expressions and to demonstrate the performance gains achieved by the TLL-assisted VLC system and the proposed lens-orientation schemes. The accuracy of the analytical outage probability expressions is verified through Monte Carlo simulations. $K_{\max}$ depends on the considered scenario. In our case, convergence was achieved for $K_{\max}$ values around $20$. Furthermore, comparisons are provided with baseline scenarios in which no lens is employed at the receiver and in which receiver orientation fluctuations are ignored. The simulations consider a room of dimensions $10\mathrm{m} \times 10\mathrm{m} \times 3.5\mathrm{m}$. Unless otherwise specified, the simulation parameters are set according to Table~I.

\begin{table}[!t]\vspace{-1mm}
		\caption{Simulation parameters.}
        \vspace{-1mm}
		\label{Tab:para}
		\centering\begin{tabular}{|c|c||c|c||c|c|}
			\hline
			Symbol & Value & Symbol & Value & Symbol & Value\\
			\hline
			$\rho$ & $1.5$ & $\lambda_p$ & $0.2$ & $\theta_{1/2}$ & $25^{\circ}$\\
			\hline
                $k_{\eta}$ & $0.1$ & $A_{L}$ & $1 \times 10^{-4}$ m$^{2}$ & $\Phi_{FoV}$ & $90^\circ$\\
                \hline
                $d_x$ & $4$ cm & $d_y$ & $4$ cm & $d_z$ & $4$ cm \\
                \hline
                $n_l$ & $1.33$ & $\sigma^2$ & $10^{-14}$ & $\mu_{\phi_{R}}$ & $0^{\circ}$\\
                \hline
                $r_{PD}$ & $0.75$~A/W & $R_{circ}$ & $1$ m & $\sigma_{\phi_R}$ & $5^{\circ}$\\
                \hline
                $\psi_x^{min}$ & $-40^{\circ}$ & $\psi_x^{max}$ & $40^{\circ}$ & $\psi_y^{min}$ & $-40^{\circ}$\\
                \hline\
                $\psi_y^{max}$ & $40^{\circ}$ & $P_i$ & $10$ W & $h$ & $3.5$ m\\
                \hline
                $R$ & $1$ m & $\epsilon_0$ & $8.854$ $\text{pFm}^{-1}$ & $\epsilon_1$ & $20$\\
                \hline
                $\gamma_{LV}$ & $22$ mN/m & $\gamma$ & $2$ $\text{Fm}^{-1}$ & $\lambda_m$ & $0.8$\\
                \hline		
                \end{tabular}
                \vspace{-1mm}
	\end{table}
    
\begin{figure}[!t]
    \centering
    \includegraphics[width=0.7\columnwidth]{"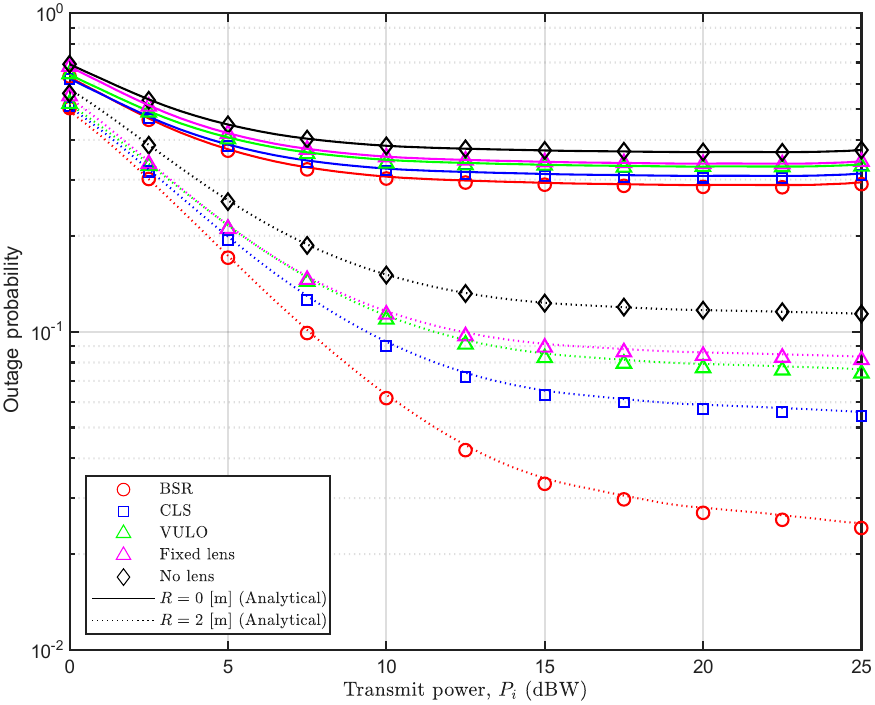"}
    \vspace{-3mm}
    \caption{Outage probability vs. transmit power $P_i$ for different values of $R$.}
    \vspace{-6mm}
    \label{fig:result0}
\end{figure}

Fig.~\ref{fig:result0} illustrates the outage probability as a function of the downlink transmit power, $P_i$, for different TLL orientation schemes. The outage probability of the proposed system is evaluated using the derived closed-form approximate expressions and is compared with Monte Carlo simulation results. The results are presented for different values of the MHCPP thinning radius, $R$. The presented results verify the accuracy of the analysis and demonstrate a tight agreement between the analytical expressions and simulation results. It is observed that the BSR scheme achieves the best outage performance among all considered schemes. This is because the BSR scheme can effectively focus the received optical signal from the tagged AP while suppressing interference from neighboring APs. Nevertheless, the low-complexity CLS and VULO schemes still provide considerable performance gains compared with the fixed-lens and no-lens scenarios. Moreover, at high transmit power levels, the outage probability curves exhibit an error floor due to the interference-limited nature of the system, where the received interference power becomes dominant. In addition, increasing the MHCPP thinning radius, $R$, significantly improves the outage performance. In particular, at $P_i = 20$ dBW and $R = 2$ m, the BSR scheme achieves approximately a $10$ dB outage reduction compared with the conventional VLC system without a lens, whereas the performance gain is only about $0.1$ dB when $R = 0$ m. This behavior arises because a larger thinning radius increases the separation between neighboring APs, thereby reducing interference. In contrast, small values of $R$ result in severe interference levels that are difficult to mitigate even with the proposed TLL design.

\begin{figure}[!t]
    \centering
    \includegraphics[width=0.7\columnwidth]{"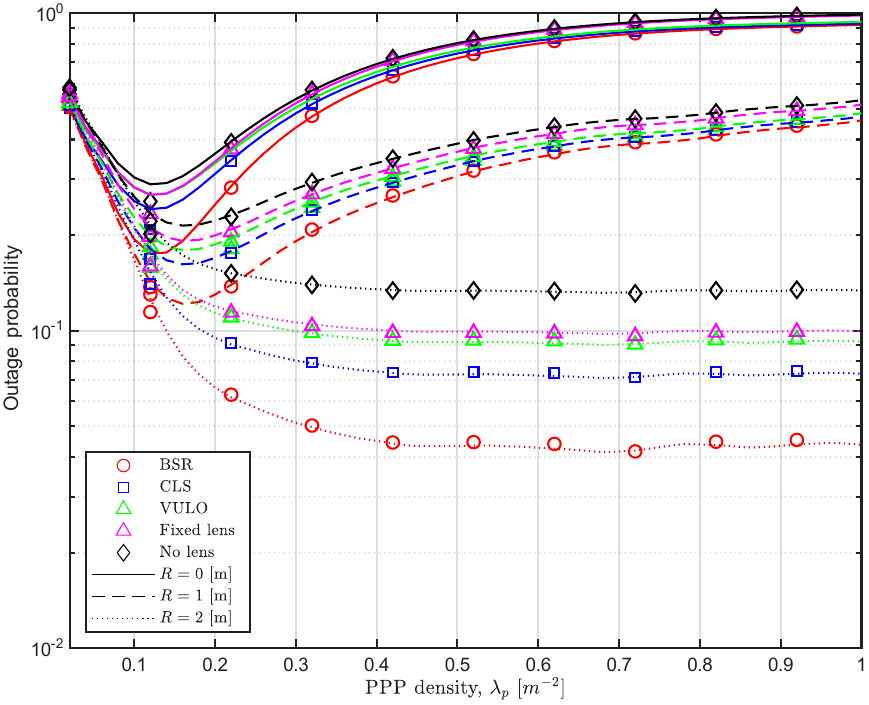"}\vspace{-3mm}
    \caption{The outage probability vs. the PPP density for APs, $\lambda_p$ for different lens orientation schemes. $P_i = 10$ W, $h = 3.5$ m, $\gamma = 2$.}
        \vspace{-6mm}
    \label{fig:result1}
\end{figure}

Fig. \ref{fig:result1} depicts the outage probability as a function of the parent PPP density used to generate AP density, $\lambda_p$, for different lens-orientation schemes. The results show that the BSR scheme consistently achieves the best outage performance among all considered schemes. Moreover, the CLS and VULO schemes also provide noticeable performance gains compared with the fixed-lens and no-lens scenarios. At low values of $\lambda_p$, the outage probability is relatively high because the APs are sparsely distributed, resulting in insufficient coverage for the user. It is further observed that the minimum separation distance between APs, characterized by the MHCPP hard-core radius $R$, has a significant impact on the system performance. In particular, larger values of $R$ lead to improved outage performance at high AP densities. For small hard-core distances, such as $R=0$, increasing $\lambda_p$ initially decreases the outage probability due to improved coverage. However, beyond a certain AP density, the outage probability starts increasing and eventually saturates at a relatively high level. This behavior is caused by the strong interference experienced in dense deployments with small AP separation distances. Consequently, an optimal AP density exists for low values of $R$. In particular, the optimal value of $\lambda_p$ for $R=0$ m and $R=1$ are approximately $1.25$ and $1.7$, respectively. In contrast, for larger hard-core distances, increasing $\lambda_p$ continuously improves the outage performance, after which the outage probability saturates at a low value. Furthermore, the outage performance for $R=2$ m is consistently better than that achieved for $R=0$ m. In addition, the performance gain obtained using the BSR scheme becomes more pronounced at larger values of $R$. Specifically, at $\lambda_p = 0.5$, the BSR scheme reduces the outage probability from $1.3\times10^{-1}$ to $4.3\times10^{-2}$ for $R=2$ m. This improvement is mainly attributed to the lower interference levels observed at larger AP separations. 

\begin{figure}[!t]
    \centering
    \includegraphics[width=0.7\columnwidth]{"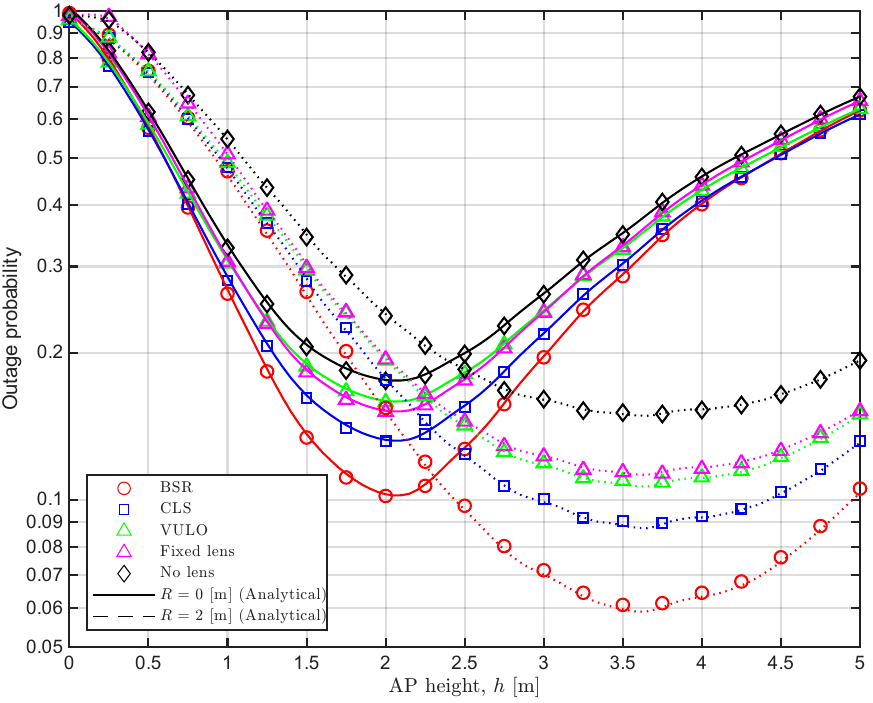"}
    \vspace{-3mm}
    \caption{Outage probability vs. the AP height, $h$, for different lens-orientation schemes. $P_i = 10$ W, $\lambda_p =0.2~\mathrm{m}^{-2}$, and $\gamma = 2$.}
    \vspace{-2mm}
    \label{fig:result2}
\end{figure}

Fig. \ref{fig:result2} illustrates the outage probability versus the AP height,  $h$. The outage probability is close to one at extremely low and high AP heights, due to the extremely low channel gains observed at both conditions. The results indicate the existence of an optimal AP height that minimizes the outage probability for each $R$ value. In particular, the optimal $h$ for $R=0$ m, and $R =2$ m are around $2$ m, and $3.5$ m, respectively. From our simulations, it is well noted that the range of $h$ which a fair outage probability can be observed, increases as $R$ increases. On the other hand, minimum outage probability corresponding to optimal $h$ increases with increased $R$. For most $R$ values, all the proposed schemes including BSR, CLS, and VULO schemes outperform the fixed lens case. However, for small $R$ like $R=0$ m, VULO scheme has worse performance than fixed lens scheme, in specific range of $h$ values. In particular, for the considered setup this range is $1.2$ to $3.2$. 

\begin{figure}[!t]
    \centering
    \includegraphics[width=0.7\columnwidth]{"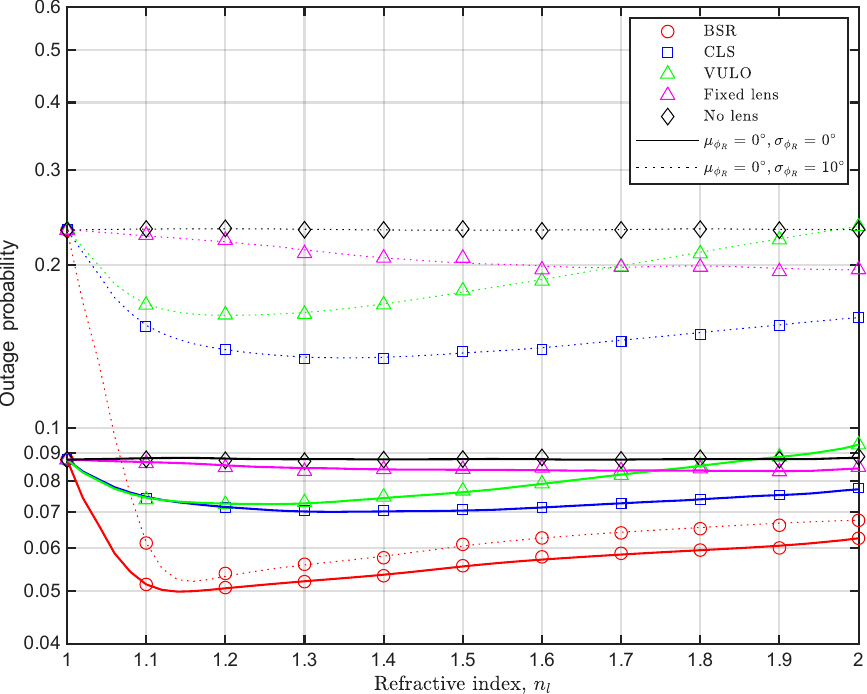"}
        \vspace{-3mm}
    \caption{Outage probability versus the liquid refractive index, $n_l$, for the proposed lens orientation schemes under different receiver orientation conditions.}
    \vspace{-2mm}
    \label{fig:result3}
\end{figure}

Fig.~\ref{fig:result3} illustrates the variation of the outage probability of the proposed TLL schemes with the refractive index of the liquid used in the TLL, \(n_l\). The range of \(n_l\) is selected from 1 to 2, which represents practical values for most liquids that can be used in TLLs. In addition, the results are presented for different values of the receiver polar angle variation, \(\sigma_{\phi_R}^{2}\). It can be clearly observed that the outage performance improves for all proposed schemes compared to the no/fixed lens benchmarks, except at \(n_l = 1\). At \(n_l = 1\), no performance gain is achieved because no light refraction occurs at the liquid surface, rendering the liquid lens ineffective. Furthermore, the results show that the BSR scheme significantly reduces the outage probability compared to the no/fixed lens schemes, particularly at high values of \(\sigma_{\phi_R}^{2}\), since it can effectively align the desired signal even under severe receiver orientation errors. Specifically, the outage probability is reduced from \(2.3\times10^{-1}\) to \(5\times10^{-2}\) at \(\mu_{\phi_R}=0^{\circ}\), \(\sigma_{\phi_R}=10^{\circ}\), and \(n_l=1.3\). For all proposed schemes, the outage probability initially decreases with increasing \(n_l\), and then gradually increases beyond a certain point. This behavior occurs because excessively high refractive indices can enhance the interference power relative to the desired signal. Consequently, each lens-orientation scheme exhibits an optimal value of \(n_l\). In particular, at \(\mu_{\phi_R}=0^{\circ}\) and \(\sigma_{\phi_R}=10^{\circ}\), the optimal values of \(n_l\) are 1.15, 1.35, and 1.2 for the BSR, CLS, and VULO schemes, respectively. Moreover, it is evident that at very high values of \(n_l\), the CLS and VULO schemes may experience limited or no performance improvement compared to the conventional no/fixed lens schemes. In particular, under severe random receiver orientation conditions, the VULO scheme performs worse than the conventional no/fixed lens receiver when \(n_l>1.8\). This behavior is attributed to the fact that the VULO scheme does not provide perfect alignment for the desired signal. Therefore, at high \(n_l\) values and large receiver orientation variations, the desired signal power from the tagged AP becomes comparable to the interference power. 

\begin{figure}[!t]
    \centering
    \includegraphics[width=0.7\columnwidth]{"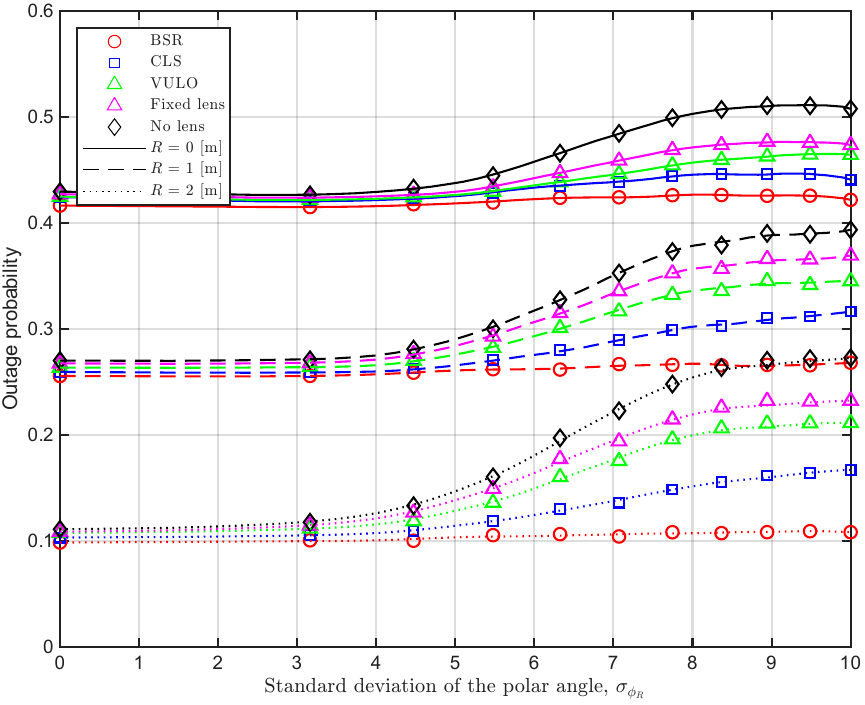"}
        \vspace{-3mm}
    \caption{Outage probability vs. the standard deviation of the polar angle, $\sigma_{\phi_R}$, under different AP heights. $P_i = 10$ W, $\lambda_p = 0.2$ $\text{m}^{-2}$, $h = 4.5$ m.}
    \vspace{-2mm}
    \label{fig:result4}
\end{figure}

Fig. \ref{fig:result4} illustrates the outage probability versus the standard deviation of the receiver polar angle, $\sigma_{\phi_R}$. The results show that the outage probability increases with increasing $\sigma_{\phi_R}$ because the LoS link between the \textit{tagged} AP and the receiver weakens, while interference becomes more significant in the large-scale network. However, the BSR scheme effectively mitigates this performance degradation, even at high $\sigma_{\phi_R}$ values, by more accurately aligning the desired signal from the \textit{tagged} AP onto the receiver PD compared to the interference signals. In particular, at $\sigma_{\phi_R}=10^{\circ}$ and $R=2$ m, the outage probability is reduced from $0.26$ to $0.11$ using the BSR scheme. Moreover, significant outage performance improvements can also be observed for the CLS and VULO schemes. Nevertheless, for all schemes, the performance improvement is limited at low $\sigma_{\phi_R}$ values. Furthermore, at $R=0$ m, the performance gain achieved by the BSR scheme is smaller than $R=2$ m case.

\vspace{-3mm}
\section{Conclusion}\label{conclusion}
This paper investigated the outage performance of large-scale TLL-assisted indoor VLC systems under random receiver orientation. An electrowetting-based TLL architecture was proposed to dynamically steer the incident optical signal toward the PD receiver while mitigating interference from neighboring APs. AP locations were modeled using an MHCPP, while receiver orientation was characterized by uniformly distributed azimuth angles and Gaussian-distributed polar angles. A tractable optical channel model was developed to jointly capture transmitter/receiver geometry, receiver orientation, and liquid lens adjustment angles. Based on this framework, three TLL orientation schemes (BSR, CLS, and VULO) were studied to trade off outage performance and implementation complexity. Using SG tools, approximate closed-form outage probability expressions were derived and validated via Monte Carlo simulations. Results show that TLL-assisted receivers significantly improve VLC reliability under orientation fluctuations and dense AP deployments. Among the schemes, BSR achieves the best performance by enhancing the desired signal and suppressing interference, while CLS and VULO provide lower-complexity gains. In particular, BSR reduces outage probability by $57.1\%$ compared with conventional fixed-lens receivers at an AP height of $3.5$ m and density of $0.2~\text{m}^{-2}$.

\vspace{-2mm}
\appendix\vspace{-1mm}
\subsection{Proof of Proposition \ref{Proposition1}}\label{Appendix1}\vspace{-1mm}
Since all received signal and interference powers are nonnegative, $\mathrm{SINR}_i>1$ implies $S_i>\sum_{k\neq i}S_k+\sigma^2$, where $S_k$ is the received signal from the $k$-th AP, which precludes any other AP from simultaneously satisfying the same condition. Under this condition, at most one AP can achieve an SINR greater than 1 and hence greater than $\gamma$ (we assume $\gamma\ge 1$). Consequently,\vspace{-1mm}
\begin{align}
    \mathbbm{1}(\text{SINR}_{0}(\mathbf{y})>\gamma) = \displaystyle\sum\nolimits_{\mathbf{x}_i\in\Phi_m}\mathbbm{1}(\text{SINR}_{i}(\mathbf{y})>\gamma).
\end{align}
Hence, the outage probability can be expressed as\vspace{-1mm}
\begin{align}
    P_o(\gamma,\mathbf{y}) = 1- \mathbb{E}\bigg[\displaystyle\sum\nolimits_{\mathbf{x}_i\in\Phi_m}\mathbbm{1}(\text{SINR}_{i}(\mathbf{y})>\gamma)\bigg].
\end{align}
By applying Campbell's Theorem~\cite{Singh_2021} and Slivnyak's Theorem~\cite{Liang_2018}, this expression reduces to
\begin{align}
    &\hspace{-0.6mm}P_o(\gamma,\mathbf{y}) \hspace{-0.6mm}=\hspace{-0.6mm} 1\hspace{-0.6mm}-\hspace{-0.6mm} \lambda_m \hspace{-0.9mm}\int_{\mathcal{S}(0,X_m,Y_m)} \hspace{-4mm} \mathbb{P}\left[\hspace{-0.6mm}\frac{\textit{l}(\mathbf{x},\mathbf{y})}{I(\Phi_m \backslash \mathbf{x})+\sigma^2}\hspace{-0.6mm}>\hspace{-0.6mm}\gamma\hspace{-0.6mm}\right] d\mathbf{x}\\
    &= \hspace{-0.6mm}1\hspace{-0.8mm}-\hspace{-0.8mm} \lambda_m \hspace{-0.9mm} \int_{\mathcal{S}(0,X_m,Y_m)} \hspace{-13mm} \mathbb{P}\big[I'\hspace{-0.6mm}<\hspace{-0.6mm}\gamma^{-1}\text{c}^2\phi_{0,j}(||\mathbf{x}\hspace{-0.6mm}-\hspace{-0.6mm}\mathbf{y}||^2
      \hspace{-0.6mm}+\hspace{-0.6mm}h_u^2)^{-(m+2)}\hspace{-0.8mm}-\hspace{-0.8mm}\sigma'^2\big] d\mathbf{x}\nonumber\\
    &= \hspace{-0.6mm}1\hspace{-0.6mm}-\hspace{-0.6mm} \lambda_m \hspace{-0.9mm}\int_{\mathcal{S}(-\mathbf{y},X_m,Y_m)} \hspace{-16mm}\mathbb{P}\big[I'<\gamma^{-1}\text{c}^2\phi_{0,j}(||\mathbf{x}||^2
    \hspace{-0.6mm}+\hspace{-0.6mm}h_u^2)^{-(m+2)}\hspace{-0.6mm}-\hspace{-0.6mm}\sigma'^2\big] d\mathbf{x},\nonumber
\end{align}
where $\mathcal{S}(-\mathbf{y},X_m,Y_m)$ denotes the shifted rooftop region obtained after the change of variables $\mathbf{x}\rightarrow\mathbf{x}+\mathbf{y}$, and $I'= I/(r_{PD}^2P_i^2k_1^2)$ represents the normalized aggregate interference power. Note that $\mathbb{P}[I'<s_p]$ is non zero only when $s_p>0$, which corresponds to $\Big\{\mathbf{x}: ||\mathbf{x}||< a_D = \Big(\Big(\frac{\gamma\sigma'^2}{\text{c}^2\phi_{0,j}}\Big)^{-\frac{1}{m+2}}-h_u^2\Big)^{1/2}\Big\}$. Accordingly, we can write\vspace{-1mm}
\begin{align}
    P_o(\gamma,\mathbf{y}) = 1- \lambda_m \int_{\mathcal{A}_{\mathbf{y}}} &\mathbb{P}\bigg[I'<\frac{\gamma^{-1}\text{c}^2\phi_{0,j}}{(||\mathbf{x}||^2+h_u^2)^{m+2}}-\sigma'^2\bigg] d\mathbf{x},
\end{align}
where $\mathcal{A}_{\mathbf{y}} = \mathcal{S}(-\mathbf{y},X_m,Y_m)\cap\mathcal{B}(0,a_D)$. Using the Gil-Peláez inversion Lemma, the cumulative density function (CDF) of the aggregate interference $I'$ in terms of its Laplace transform $\mathcal{L}_{I'}(-jt)$ is given by\vspace{-1mm}
\begin{align}
    \mathbb{P}[I'<s_p] = \frac{1}{2}-\frac{1}{\pi}\int_0^{\infty}\frac{1}{t}\text{Im}\left[e^{-jts_p}\mathcal{L}_{I'}(-jt)\right]dt.
\end{align}
Thus, the SINR outage probability can be expressed as
\begin{align}\label{P_out1}
    P_o(\gamma,\mathbf{y}) \hspace{-0.6mm}&=\hspace{-0.6mm} 1\hspace{-0.6mm}-\hspace{-0.6mm} \lambda_m \hspace{-0.8mm}\int_{\mathcal{A}_{\mathbf{y}}}\hspace{-0.8mm}\Big[\frac{1}{2}\hspace{-0.8mm}-\hspace{-0.8mm}\frac{1}{\pi}\hspace{-0.8mm}\int_0^{\infty}\hspace{-0.8mm}\frac{1}{t}\text{Im}\big[\exp\big(-jt(\gamma^{-1}\text{c}^2\phi_{0,j}\nonumber\\
    &\times\hspace{-0.6mm}(||\mathbf{x}||^2\hspace{-0.6mm}+\hspace{-0.6mm}h_u^2)^{-(m+2)}\hspace{-0.6mm}-\hspace{-0.6mm}\sigma'^2)\big]\mathcal{L}_{I'}(-jt)\big]dt\Big]d\mathbf{x}.
\end{align}
The Laplace transform of the aggregate interference is given by $\mathcal{L}_{I'}(s) = \mathbb{E}[\exp\left(-s{I'}\right)]$. We evaluate $\mathcal{L}_{I'}(s)$ using the probability-generating functional (PGFL) of a homogeneous PPP with intensity $\lambda_{m}$, while emulating the hard-core repulsion of MHCPP-Type II by imposing an exclusion (guard) region of radius $R$ around the \textit{tagged AP} (serving AP)~\cite{Hesham_2017}. This is motivated by the need to reduce the performance gap when approximating an MHCPP-Type II by a PPP. Accordingly, $\mathcal{L}_{I'}(s)$ can be expressed as\vspace{-1mm}
 \begin{align}   
    \hspace{-2.5mm}\mathcal{L}_{I'}(s)\hspace{-0.7mm}= \hspace{-0.7mm}\exp\hspace{-1mm}\bigg(\hspace{-1mm}-\lambda_m \hspace{-2mm}\int_{\mathcal{A}_{\mathbf{x, y}}}\hspace{-5mm}1\hspace{-1mm}-\hspace{-1mm}\exp\left(-s\textit{l}(\mathbf{z}, 0)/(r_{PD}^2P_i^2k_1^2)\right)\hspace{-0.5mm}d\mathbf{z}\hspace{-0.5mm}\bigg)\hspace{-0.8mm},
\end{align}
where $\mathcal{A}_{\mathbf{x, y}} = \mathcal{S}_y\backslash \mathcal{B}(\mathbf{x},R)$, $\mathcal{S}_y=\mathcal{S}(-\mathbf{y},X_m,Y_m)$, and $\mathcal{B}(\mathbf{x},R)$ is the guard region. Then, \eqref{P_out1} can be rewritten as\vspace{-1mm}
\begin{align}\label{P_out2}
    &P_o(\gamma,\mathbf{y}) = \frac{2-\lambda_m|\mathcal{A}_{\mathbf{y}}|}{2}+ \frac{\lambda_m}{\pi} \int_0^{\infty}\frac{1}{t}\times\\
    &\text{Im}\bigg[\int_{\mathcal{A}_{\mathbf{y}}}\exp\left(-jt(\gamma^{-1}\text{c}^2\phi_{0,j}(||\mathbf{x}||^2+h_u^2)^{-(m+2)}-\sigma'^2)\right)\nonumber \\
    &\times\exp\hspace{-1mm}\bigg(\hspace{-1mm}-\lambda_m\hspace{-1mm}\int_{\mathcal{A}_{\mathbf{x, y}}}\hspace{-1mm}\bigg(1-\exp\hspace{-1mm}\bigg[\frac{-jt\text{c}^2\phi_{\mathbf{z},0}}{(||\mathbf{z}||^2\hspace{-1mm}+\hspace{-1mm}h_u^2)^{m+2}}\bigg]\bigg)d\mathbf{z}\bigg)d\mathbf{x}\bigg]dt.\nonumber\vspace{-1mm}
\end{align}
Since the interference field is approximately stationary, and mainly dependent only on $R$ and not $\mathbf{x}$, the dependence of the exclusion-region contribution on the serving AP location $\mathbf{x}$ is neglected for analytical tractability. Hence, the integral region for the interference, $\mathcal{A}_{\mathbf{x,y}}$ is replaced with $\bar{\mathcal{A}}_{\mathbf{y}} = \mathcal{S}_y\backslash \mathcal{B}(\mathbf{0},R)$ to obtain the general expression for $P_o(\gamma,\mathbf{y})$.\vspace{-1mm}

\subsection{Proof of Proposition \ref{Proposition2}}\label{Appendix2}\vspace{-1mm}
Under the condition $\text{c}(\phi_{i,0})=1$, $\mathcal{F}(-jt,\mathcal{A}_{\mathbf{y}})$ reduces to\vspace{-1mm}
\begin{align}\label{f11}
    \hspace{-2mm}\mathcal{F}(-jt,\mathcal{A}_{\mathbf{y}}) \hspace{-0.5mm}\hspace{-0.5mm}=\hspace{-0.9mm} \int_{\mathcal{A}_{\mathbf{y}}} \hspace{-0.9mm}\exp\left(-jt\gamma^{-1}(||\mathbf{x}||^2\hspace{-0.5mm}+\hspace{-0.5mm}h_u^2)^{-(m+2)}\right)\hspace{-0.5mm}d\mathbf{x}.
\end{align}
By adopting the practical assumption $a_D \le \text{min}\left\{\frac{X_m}{2}, \frac{Y_m}{2}\right\}$, the integral \eqref{f11} is expressed in cylindrical coordinates as 
\begin{align}\label{d_in1}
    \hspace{-3.6mm}\mathcal{F}(\hspace{-0.5mm}-jt,\mathcal{A}_{\mathbf{y}}\hspace{-0.5mm}) \hspace{-0.5mm}=\hspace{-1.5mm} \int_{0}^{2\pi} \hspace{-3mm}\int_{0}^{a_D} \hspace{-4mm}\exp\hspace{-0.5mm}\left(\hspace{-0.8mm}-jt\gamma^{-1}\hspace{-0.5mm}(r^2\hspace{-1.2mm}+\hspace{-0.8mm}h_u^2)^{-(m+2)}\hspace{-0.5mm}\right)\hspace{-0.5mm}rdr d\theta_x.
\end{align}
Exploiting the fact that the inner integral in \eqref{d_in1} is independent of $\theta_x$, and using the variable transformation $u = r^2 + h_u^2$, the integral in \eqref{d_in1} is further simplified to \vspace{-1mm}
\begin{align}\label{d_in3}
    \hspace{-3.1mm}\mathcal{F}(-jt,\mathcal{A}_{\mathbf{y}}) = \hspace{-0.5mm}\pi\int_{h_u^2}^{a_D^2+h_u^2} \exp\left(-jt\gamma^{-1}u^{-(m+2)}\right)du.
\end{align}
Next, we use the complex exponential series expansion to expand the integrand in \eqref{d_in3}, yielding\vspace{-1mm}
\begin{align}\label{d_in4}
    \hspace{-3.1mm}\mathcal{F}(-jt,\mathcal{A}_{\mathbf{y}}) = \hspace{-0.5mm}\pi\int_{h_u^2}^{a_D^2+h_u^2} 
    \sum_{k=0}^{\infty} \frac{(-jt\gamma^{-1})^k}{k!}u^{-(m+2)k}
    du.
\end{align}
We note that for $u\in [h_u^2,a_D^2+h_u^2]$, $|u^{-(m+2)k}|\le h_u^{-2(m+2)k}$. Hence, $|\frac{(-jt\gamma^{-1})^k}{k!}u^{-(m+2)k}|\le \frac{|t\gamma^{-1}|^k}{k!}h_u^{-2(m+2)k}$, and $\sum_{k=0}^{\infty}\frac{|t\gamma^{-1}|^k}{k!}h_u^{-2(m+2)k} = \exp{|t\gamma^{-1}|h_u^{-2(m+2)}}$. Thus, the series is uniformly bounded by an integrable constant, and by dominated convergence or Weierstrass M-test~\cite{George_2012}, the sum and integration in~\eqref{d_in4} can be interchanged. Using $\int u^{-(m+2)k} = \frac{u^{1-(m+2)k}}{1-(m+2)k}$, the final expression is derived.\vspace{-1mm}

\subsection{Proof of Proposition \ref{Proposition3}}\label{Appendix3}\vspace{-1mm}
We note that the light incident angle at the PD from interfering APs varies with the location of the interfering AP, the surface orientation of the liquid lens, and the receiver orientation. For the BSR scheme, interference rays experience strong angular misalignment after refraction via the liquid lens. Hence, we assume $\frac{
t\,c^2(\phi_{\mathbf z,0})}{(\|\mathbf z\|^2+h_u^2)^{m+2}}
\ll 1$ and the approximation $e^{-x}\approx 1-x$ can be employed. Then, \eqref{K_SA} reduces to\vspace{-1mm}
    \begin{equation}
    \hspace{-2mm}\mathcal K(-jt,\bar{\mathcal A}_{\mathbf y})
    \approx
    \exp\!\bigg(
    -\lambda_m jt
    \int_{\bar{\mathcal A}_{\mathbf y}}
    \frac{
    c^2\phi_{\mathbf z,0}
    }{
    (\|\mathbf z\|^2+h_u^2)^{m+2}
    }
    d\mathbf z
    \bigg).
    \label{K_first_order}
    \end{equation}
Due to the angular misalignment of interference signals, we use $\mathbb E[\text{c}^2\phi_{\mathbf z,0}]$ instead of individual incident angles. Using the refraction model in \eqref{cos_theta}, exploiting the weak interferer-alignment property of the BSR scheme, averaging over receiver orientations, and exploiting isotropy yields
    \begin{equation*}
    \hspace{-1.4mm}\bar G_{\mathrm{a}}
    \hspace{-0.5mm}=\hspace{-0.5mm}
    \mathbb E[c^2\phi_{\mathbf z,0}]
    \hspace{-0.5mm}\approx\hspace{-0.5mm}
    \frac1{n_l^2}\hspace{-0.5mm}
    \bigg[\hspace{-0.5mm}
    \frac12
   \hspace{-0.9mm} +\hspace{-0.9mm}
    \bigg(\hspace{-0.5mm}
    n_l^2\hspace{-0.7mm}-\hspace{-0.7mm}\frac23\hspace{-0.5mm}
    \bigg)\hspace{-0.9mm}
    \frac{
    1\hspace{-0.5mm}+\hspace{-0.5mm}
    \exp(-2\sigma_{\phi_R}^2)
    \text{c}(2\mu_{\phi_R})
    }{2}
    \bigg]\hspace{-1mm}.
    \end{equation*}
Substituting the above expression into \eqref{K_first_order}, converting the spatial integral into polar coordinates, and using the transformation $u^2 = r^2+h_u^2$, we obtain the final expression.\vspace{-1mm}

\subsection{Proof of Theorem \ref{Theorem1}}\label{Appendix4}\vspace{-1mm}
Substituting \eqref{F_final_BSR} and \eqref{K_closed_BSR} into \eqref{eq_Pout}, and after some mathematical manipulations, the outage probability under the BSR scheme can be approximated as\vspace{-1mm}
    \begin{equation}
    P_o^{\rm BSR}(\gamma,\mathbf y)
    \hspace{-1mm}\approx\hspace{-1mm}
    1\hspace{-1mm}-\hspace{-1mm}\frac{\lambda_m|\mathcal A_{\mathbf y}|}{2}
    \hspace{-1mm}+\hspace{-1mm}
    \lambda_m
    \hspace{-1mm}\sum_{k=1}^{K_{\max}}\hspace{-1mm}
    C_k
    \frac{
    \Gamma(k)
    \text{s}\left(\frac{k\pi}{2}\right)
    }{
    (\sigma'^2\hspace{-1mm}-\hspace{-1mm}\lambda_m\Xi_{\rm BSR})^k
    },
    \label{Po_BSR_conditional}
    \end{equation}
where $K_{\text{max}}$ is the truncation order of the infinite series expansion (its value depends on the considered scenario; see Sec. VI), and\vspace{-1mm}
    \begin{equation}
    C_k=
    \frac{
    \gamma^{-k}
    h_u^{2(1-k/k_m)}
    }{
    k!(1-k/k_m)
    }
    \bigg[
    \bigg(
    \frac{a_D^2}{h_u^2}+1
    \bigg)^{1-k/k_m}
    -1
    \bigg].
    \label{Ck_def}
        \end{equation}
Next, we substitute $a_D=((\gamma\sigma'^2)^{-\frac{1}{m+2}}-h_u^2)^{1/2}$ corresponding to the BSR scheme, into \eqref{Po_BSR_conditional}. Then, averaging over the receiver location and orientations as in \eqref{P_outf1}, and using the practical assumption $a_D\ll \min(X_m,Y_m)$, the result follows.\vspace{-1mm}

\subsection{Proof of Proposition \ref{Proposition4}}\label{Appendix5}\vspace{-1mm}
    In the CLS scheme, using the relation $\text{c}\phi_{i,0} \hspace{-0.5mm}=\hspace{-0.5mm} \hat{\boldsymbol{\eta}}_{RT,i,0}\hspace{-0.5mm}\cdot \hspace{-0.5mm}\hat{\boldsymbol{\eta}}_{R,0}$, the term $\text{c}\phi_{i,0}$ can be expanded as\vspace{-1mm}
\begin{align*}
    \text{c}\phi_{i,0} \hspace{-0.5mm}=\hspace{-0.5mm} e_{RT,x,i,0}\text{c}\theta_R \text{s}\phi_R \hspace{-0.5mm}+\hspace{-0.5mm} e_{RT,y,i,0}\text{s}\theta_R\text{s}\phi_R \hspace{-0.5mm}+\hspace{-0.5mm}e_{RT,z,i,0}\text{c}\phi_R.
\end{align*}
Substituting $\text{c}\phi_{i,0}$ into the expression for $\mathcal{F}(-jt,\mathcal{A}_{\mathbf{y}})$ yields\vspace{-1mm}
\begin{align}
    &\mathcal{F}(-jt,\mathcal{A}_{\mathbf{y}}) = \int_{\mathcal{A}_{\mathrm{y}}} \exp\{-jt\gamma^{-1}(e_{RT,x,i,0}\text{c}\theta_R \text{s}\phi_R\\
    &+ e_{RT,y,i,0}\text{s}\theta_R\text{s}\phi_R \hspace{-0.7mm}+\hspace{-0.5mm}e_{RT,z,i,0}\text{c}\phi_R)^2\hspace{-0.5mm}\hspace{-0.1mm}(||\mathbf{x}||^2\hspace{-0.5mm}+\hspace{-0.5mm}h_u^2)^{-(m+2)}\}d\mathbf{x}. \nonumber
\end{align}
In cylindrical coordinates, the above expression becomes
\begin{align}
    &\mathcal{F}(-jt,\mathcal{A}_{\mathbf{y}}) = \int_{0}^{2\pi} \int_{0}^{a_D}  \exp\{-jt\gamma^{-1}(r\text{c}\theta_x\text{c}\theta_R \text{s}\phi_R \\
    &+ r\text{s}\theta_x\text{s}\theta_R\text{s}\phi_R +h_u\text{c}\phi_R)^2(r^2+h_u^2)^{-(m+3)}\}rdrd\theta_x. 
    \nonumber
    \label{ghjkhf}
\end{align}
To evaluate the integral, we first integrate over $\theta_x$, exploiting its independence from $r$. Define $A_c = r\text{c}\theta_R\text{s}\phi_R$, $B_c = r\text{s}\theta_R\text{s}\phi_R$, $C_c = h_u\text{c}\phi_R$, and $R_c = \sqrt{A_c^2+B_c^2} = r\text{s}\phi_R$. Using the identity $A_c \cos\theta_x + B_c \sin\theta_x = R_c \cos(\theta_x - \theta_R)$, together with standard trigonometric relations, the integral over $\theta_x$ can be simplified as\vspace{-1mm}
\begin{align}
    &I_1 \hspace{-1mm}= \hspace{-0.5mm}\exp\{\hspace{-0.5mm}-jt\gamma^{-1}C_c^2(r^2\hspace{-0.5mm}+\hspace{-0.5mm}h_u^2)^{\hspace{-0.5mm}-(m\hspace{-0.5mm}+\hspace{-0.5mm}3)}\}
    \hspace{-1mm}\int_{0}^{2\pi} \hspace{-3.5mm}\exp\{\hspace{-0.5mm}-jt\gamma^{-1}\hspace{-0.5mm}\\
    &\times(r^2\hspace{-1mm}+\hspace{-0.5mm}h_u^2)^{\hspace{-0.5mm}-(m\hspace{-0.5mm}+\hspace{-0.5mm}3)}[R_c^2\text{c}^2(\theta_x\hspace{-0.5mm}-\hspace{-0.5mm}\theta_R)+2R_cC_c\text{c}(\theta_x\hspace{-0.5mm}-\hspace{-0.5mm}\theta_R)]\}d\theta_x.\nonumber
\end{align}
Since the integration is periodic, it is independent of $\theta_R$. Applying the identity $\text{c}^2(\theta) =\frac{1+\text{c}(2\theta)}{2}$, and using the approximation $\exp\{-(jt\gamma^{-1}(r^2+h_u^2)^{-(m+3)}R_c^2/2)\text{c}(2\theta_x)\} \approx 1$, which is valid for $t(r^2+h_u^2)^{-(m+3)}R_c^2 \ll 1$, together with the Bessel identity $\int_{0}^{2\pi}\exp(x\text{c}(\theta))d\theta = 2\pi I_0(x)$, the integration in $I_1$ can be approximated as 
\begin{align}
    \label{eq:I1_approx}
    I_1 &\approx 2 \pi \,
    \exp\Big[-j t \gamma^{-1}(r^2 + h_u^2)^{-(m+3)} \Big(C_c^2 + \frac{R_c^2}{2}\Big)\Big] \nonumber \\
    &\times I_0 \Big( 2 j t \gamma^{-1}\, R_c C_c \, (r^2 + h_u^2)^{-(m+3)} \Big).
\end{align}
Next, we evaluate $\int_0^{2\pi}I_1rdr$ to obtain $\mathcal{F}(-jt,\mathcal{A}_{\mathbf{y}})$. Since $R_c$ depends on $r$, we substitute $R_c = r\text{s}\phi_R$. Using the variable transformation $u = r^2+h_u^2$, together with the series representation of the Bessel function, $I_0(x) = \sum_{k=0}^{\infty} \frac{1}{(k!)^2}\left(\frac{x}{2}\right)^{2k}$~\cite{George_2012}, we obtain
\begin{align}
    \hspace{-1mm}&\mathcal{F}(-jt,\mathcal{A}_{\mathbf{y}}) \hspace{-0.5mm}\approx\hspace{-0.5mm} \pi \hspace{-0.5mm}\sum_{k=0}^{\infty}\frac{(jt\gamma^{-1}\text{s}\phi_R C_c)^{2k}}{(k!)^2}\hspace{-2.5mm}\int_{h_u^2}^{a_D^2+h_u^2}\hspace{-7mm} (u\hspace{-0.5mm}-\hspace{-0.5mm}h_u^2)^ku^{-2k(m+3)}\nonumber \\
    &\times \exp\hspace{-0.5mm}\left(\hspace{-0.5mm}-jt\gamma^{-1}(C_c^2+(u\hspace{-0.5mm}-\hspace{-0.5mm}h_u^2)\text{s}^2\phi_R/2)u^{-(m+3)}\right) \hspace{-0.5mm}du. 
    \label{equ154}
\end{align}
To further simplify the above expression, we use a representative squared radial distance in the serving region, namely $a_D^2/2$. Hence, the approximation $u-h_u^2\approx a_D^2/2$ is employed in \eqref{equ154}. Then, using the variable transformation $w = jt\gamma^{-1}(C_c^2+a_D^2\text{s}^2\phi_R/4)u^{-(m+3)}$, the integral can be evaluated in closed form.\vspace{-2mm}

\subsection{Proof of Proposition \ref{Proposition5}}\label{Appendix6}\vspace{-1mm}
Once the lens is aligned toward the tagged AP, interfering rays impinge on the PD after passing through a lens surface that is misaligned with respect to the interference links. Consequently, the lens no longer provides deterministic beam steering for the interference signals. Instead, the received interference power is primarily determined by the alignment between the refracted interference ray and the PD surface normal. Hence, using the relations $\hat{\boldsymbol{\eta}}_{L,j}=\hat{\boldsymbol{\eta}}_{RT,0,j}$, and $\text{c}({\alpha}_{\mathbf{z},j}) = \hat{\boldsymbol{\eta}}_{RT,0,j}\cdot \hat{\boldsymbol{\eta}}_{RT,\mathbf{z},j}$, the term $\mathcal{K}(-jt,\bar{\mathcal{A}}_{\mathbf{y}})$ in the CLS scheme can be expressed as
\begin{align}
    \mathcal{K}(-jt,\bar{\mathcal{A}}_{\mathbf{y}}) &\hspace{-0.5mm}=\hspace{-0.5mm} \text{exp}\bigg(\hspace{-1.5mm}-\lambda_m\hspace{-0.5mm}\int_{\bar{\mathcal{A}}_{\mathbf{y}}} \hspace{-3mm}\bigg(\hspace{-1mm}1\hspace{-0.5mm}-\hspace{-0.5mm}\exp\bigg(\hspace{-1.5mm}\frac{-jt\text{c}^2\phi_{\mathbf{z},j}}{(||\mathbf{z}||^2\hspace{-0.5mm}+\hspace{-0.5mm}h_u^2)^{m+2}}\bigg)\hspace{-1.5mm}\bigg)\hspace{-0.5mm}d\mathbf{z}\hspace{-0.5mm}\bigg)\hspace{-0.5mm},
    \label{equ:K_CLS4}
\end{align} 
where $\text{c}\phi_{\mathbf{z},j}$ is given by
\begin{align}\label{cos_theta_CLS2}
        &\text{c}\phi_{\mathbf{z},j} \hspace{-0.5mm}= \hspace{-0.5mm}-{n_{l}}^{-1}\Big[\big(\big(\hat{\boldsymbol{\eta}}_{RT,0,j}\hspace{-0.5mm}\cdot\hspace{-0.5mm} \hat{\boldsymbol{\eta}}_{RT,\mathbf{z},j}\hspace{-0.5mm}-\hspace{-0.5mm}(n_l^2 \hspace{-0.5mm}- \hspace{-0.5mm}1+\hspace{-0.5mm} \\
        &\hspace{-0.5mm}\hspace{-0.5mm}(\hat{\boldsymbol{\eta}}_{RT,0,j}\hspace{-0.5mm}\cdot\hspace{-0.5mm}\hat{\boldsymbol{\eta}}_{RT,\mathbf{z},j})^2)^{\frac{1}{2}}\big)(\hat{\boldsymbol{\eta}}_{RT,0,j}\hspace{-0.5mm}\cdot\hspace{-0.5mm} \hat{\boldsymbol{\eta}}_{R,j})\big)\hspace{-0.5mm}
        -\hspace{-0.5mm}(\hat{\boldsymbol{\eta}}_{RT,\mathbf{z},j}\cdot\hat{\boldsymbol{\eta}}_{R,j})\big)\Big]. \nonumber
\end{align}
First, we simplify \eqref{cos_theta_CLS2} using the definitions of $\hat{\boldsymbol{\eta}}_{RT,0,j}$, $\hat{\boldsymbol{\eta}}_{RT,\mathbf{z},j}$, and $\hat{\boldsymbol{\eta}}_{R,j}$. Next, we substitute \eqref{cos_theta_CLS2} into \eqref{equ:K_CLS4}, and convert the integral in~\eqref{equ:K_CLS4} to cylindrical coordinates to obtain
\begin{align}
    &\mathcal{K}(-jt,\bar{\mathcal{A}}_{\mathbf{y}})
    =
    \exp\!\bigg(
    -\pi\lambda_m(a_D^2-R^2)\\
    &+\lambda_m
    \int_R^{a_D}
    \int_0^{2\pi}
    \exp\!\bigg(
    \frac{-jt\,\text{c}_{\mathbf{z},j}^2(r,\theta_x)}
    {(r^2+h_u^2)^{m+2}}
    \bigg)
    r\, d\theta_x\, dr
    \bigg), \nonumber 
\end{align}
where 
\begin{align*}
    \hspace{-2.9mm}\text{c}_{\mathbf{z},j}(r,\theta_x)
    \hspace{-1mm}&=\hspace{-1mm}
    -\frac{1}{n_l}\hspace{-1mm}
    \bigg[\hspace{-0.5mm}
    \kappa(r,\theta_x)\hspace{-0.3mm}C_0(r)
    \hspace{-1mm}-\hspace{-1mm}
    \frac{
    r \text{s}\phi_R \text{c}(\theta_x\hspace{-1mm}-\hspace{-1mm}\theta_R)
   \hspace{-1mm} +\hspace{-1mm}
    h_u \text{c}\phi_R
    }
    {\sqrt{r^2+h_u^2}}\hspace{-1mm}
    \bigg]\hspace{-1mm},\\
    C_0(r) &= \frac{
    r_0r\text{c}(\theta_0-\theta_R)+h_u^2
    }
    {\sqrt{(r_0^2+h_u^2)(r^2+h_u^2)}},
\end{align*}
and
\begin{equation*}
    \kappa(r,\theta_x)
    \hspace{-0.9mm}=\hspace{-0.9mm}
    \frac{
    r_0r\text{c}(\theta_x\hspace{-0.8mm}-\hspace{-0.7mm}\theta_0)\hspace{-0.7mm}+\hspace{-0.7mm}h_u^2
    }
    {\hspace{-0.5mm}\left(\hspace{-0.5mm}(r_0^2\hspace{-0.7mm}+\hspace{-0.7mm}h_u^2)(r^2\hspace{-0.6mm}+\hspace{-0.6mm}h_u^2)\hspace{-0.5mm}\right)^{\frac{1}{2}}}
    -
    \Big[\hspace{-0.5mm}
    n_l^2\hspace{-0.5mm}-\hspace{-0.5mm}1\hspace{-0.5mm}+\hspace{-0.5mm}
    \frac{
    \left(r_0r\text{c}(\theta_x\hspace{-0.7mm}-\hspace{-0.7mm}\theta_0)\hspace{-0.7mm}+\hspace{-0.7mm}h_u^2\right)^2
    }
    {(r_0^2\hspace{-0.5mm}+\hspace{-0.5mm}h_u^2)(r^2\hspace{-0.5mm}+\hspace{-0.5mm}h_u^2)}\hspace{-0.5mm}
    \Big]^{\hspace{-1mm}\frac{1}{2}}\hspace{-2mm}.
\end{equation*}
For interference signals, $\frac{t\,\mathrm{c}_{\mathbf{z},j}^2(r,\theta)}{(r^2+h_u)^2)^{m+2}} \hspace{-0.65mm}\ll \hspace{-0.65mm}1$. Hence, using the first-order Taylor series approximation $\exp(-x)\approx 1-x$, and subsequently applying the identities $\int_{0}^{2\pi}\text{c}^2(\theta_x-\theta_R)d\theta_x=\pi$, and $\int_{0}^{2\pi}\text{c}(\theta_x-\theta_R)d\theta_x=0$, we obtain
\begin{equation}\label{equ:lms123}
        \mathcal{K}(\hspace{-0.5mm}-jt,\bar{\mathcal{A}}_{\mathbf{y}})
        \hspace{-0.5mm}\approx\hspace{-0.5mm}
        \exp\!\hspace{-0.5mm}\bigg(\hspace{-1mm}
        -\frac{\lambda_m jt}{n_l^2}\hspace{-2.5mm}
        \int_R^{a_D}\hspace{-1.5mm}
        \frac{\Xi_C(r)
        \hspace{-1mm}+\hspace{-1mm}
        \frac{
        \pi r^2 \text{s}^2\phi_R
        +
        2\pi h_u^2 \text{c}^2\phi_R
        }
        {r^2+h_u^2}
        }
        {(r^2+h_u^2)^{m+2}}
        \hspace{-0.5mm} r dr\hspace{-1.5mm}
        \bigg)\hspace{-1mm},
    \end{equation}
where $\Xi_C(r)=C_0^2(r)\int_{0}^{2\pi}\kappa^2(r,\theta_x)d\theta_x$. Since $\kappa^2(r,\theta_x)$ depends on $\theta_x$ only through $\text{c}(\theta_x-\theta_0)$, the cosine averages out to zero over one period. Hence, the integral in $\Xi_C(r)$ can be approximated as $I_{\Xi_C(r)}\approx 2\pi(A_0(r)-(n_l^2-1+A_0^2(r))^{\frac{1}{2}})^2$, where $A_0(r) = h_u^2/((r_0^2+h_u^2)(r^2+h_u^2))^{1/2}$. To evaluate the first term in~\eqref{equ:lms123}, we use the approximation $A_0(r)\ll n_l^2$, apply the first-order Taylor series expansion to $\sqrt{n_l^2-1+A_0^2(r)}$, and then use the variable transformation $u=r^2+h_u^2$. The second term in the integral of~\eqref{equ:lms123} can be evaluated simply using the transformation $u=r^2+h_u^2$. Combining the resulting expressions, we can write
\begin{equation}\label{equ:K_CLS112}
        \mathcal{K}(-jt,\bar{\mathcal{A}}_{\mathbf{y}})
        \hspace{-1mm}\approx\hspace{-1mm}
        \exp\!\bigg(\hspace{-1.5mm}
        -\frac{\lambda_m jt\,\pi}{2 n_l^2}\hspace{-3.5mm}
        \sum_{k=m+1}^{m+3}\hspace{-3.5mm}
        D_k\hspace{-0.7mm}
        \left[
        (R^2\hspace{-1mm}+\hspace{-1mm}h_u^2)^{-k}
        \hspace{-2mm}-\hspace{-1mm}
        (a_D^2\hspace{-1mm}+\hspace{-1mm}h_u^2)^{-k}
        \right]\hspace{-1.5mm}
        \bigg)\hspace{-1.5mm},
    \end{equation}
where $D_{m+1} = \frac{2(n_l^2-1)r_0^2\text{c}^2(\theta_0-\theta_R)}{r_0^2+h_u^2}+\text{s}^2\phi_R$, and $D_{m+2} = 2h_u^2\big[\frac{(n_l^2-1)h_u^2}{r_0^2+h_u^2}+\text{c}^2\phi_R\big]$, and $D_{m+3} = \frac{4(n_l^2-1)r_0h_u^2\text{c}(\theta_0-\theta_R)}{r_0^2+h_u^2}$.

Next, we eliminate the explicit dependence on the serving AP angular coordinates by averaging over the isotropic nearest-AP
distribution. Under the PPP approximation of the MHCPP, the nearest AP angle $\theta_0$ is uniformly distributed over $[0,2\pi)$ and independent of the receiver orientation angle
$\theta_R$. Hence, $\mathbb{E}_{\theta_0}\!\left[\text{c}(\theta_0-\theta_R)\right]=0$, and
$\mathbb{E}_{\theta_0}\!\left[\text{c}^2(\theta_0-\theta_R)\right] =\frac{1}{2}$. Also, the PDF of serving AP distance
$r_0$ is given by $f_{r_0}(r)= 2\pi\lambda_m r \exp(-\pi\lambda_m r^2)$, for $r \ge 0$. Using these identities, $\mathcal{K}(-jt,\bar{\mathcal{A}}_{\mathbf{y}})$ can be expressed by replacing $D_{m+1}$, $D_{m+2}$, and $D_{m+3}$ with $\bar{D}_{m+1}=(n_l^2-1)\left(1-\pi\lambda_m h_u^2\exp(\pi\lambda_m h_u^2)E_1(\pi\lambda_m h_u^2)\right)+ \text{s}^2\phi_R$, $\bar{D}_{m+2}=2h_u^2\left((n_l^2-1)\pi\lambda_m h_u^2\exp(\pi\lambda_m h_u^2)E_1(\pi\lambda_m h_u^2)+ \text{c}^2\phi_R\right)$, and $\bar{D}_{m+3}=0$, respectively.\vspace{-1mm}

\subsection{Proof of Theorem \ref{Theorem2}}\label{Appendix7}\vspace{-1mm}
We substitute~\eqref{equ:F_CLS1}
and~\eqref{equ:K_CLS1} into~\eqref{eq_Pout} to obtain
    \begin{align}
    \hspace{-2mm}P_o^{\mathrm{CLS}}(\gamma,\mathbf y)
    \hspace{-0.5mm}=\hspace{-0.5mm}
    \frac{2-\lambda_m|\mathcal{A}_{\mathbf{y}}|}{2}
    \hspace{-0.5mm}+\hspace{-0.5mm}
    \frac{\lambda_m}{m+3}\hspace{-0.9mm}
    \sum_{k=0}^{K_{\max}}\hspace{-0.9mm}
    \frac{
    (\text{s}{\phi_R}C_c)^{2k}
    }{
    (k!)^2
    }
    I_k ,
    \label{sdfgsgs}
    \end{align}
where $I_k=\int_0^\infty t^{2k-1}
\text{Im}\left[(j\beta t)^{\frac1{m+3}-2k}e^{jt(\sigma'^2-\Xi_{\mathrm{CLS}})}G(t)\right]dt$, $\beta = \gamma^{-1}(C_c^2+a_D^2\text{s}^2{\phi_R}/4)$, $G(t) =
\Gamma(\mu,\beta_1jt) - \Gamma(\mu,\beta_2jt)$, $\beta_1=\beta
(a_D^2+h_u^2)^{-(m+3)}$, and $\beta_2=\beta h_u^{-2(m+3)}$. Next, using the identity $(j\beta t)^{\frac1{m+3}-2k} =
(\beta t)^{\frac1{m+3}-2k} e^{j\theta_k}$, where $\theta_k = \frac{\pi}{2} \left(\frac1{m+3}-2k \right)$, together with the small-argument approximation of the incomplete Gamma function, we obtain
    \begin{align}
    P_o^{\mathrm{CLS}}(\gamma,\mathbf y)
    &\hspace{-0.8mm}\approx\hspace{-0.8mm}
    \frac{2\hspace{-0.7mm}-\hspace{-0.7mm}\lambda_m|\mathcal{A}_{\mathbf{y}}|}{2}
    \hspace{-0.7mm}+\hspace{-0.7mm}
    \frac{\lambda_m}{m+3}
    \hspace{-1.8mm}\sum_{k=0}^{K_{\max}}\hspace{-1mm}
    \frac{
    (h_u\text{s}\phi_R\text{c}\phi_R)^{2k}
    \beta^{\frac1{m+3}\hspace{-0.5mm}-\hspace{-0.5mm}2k}
    }{
    (k!)^2
    }
    \nonumber\\
    &\quad\times
    \frac{
    (\beta_2^\mu-\beta_1^\mu)\Gamma(2\mu)
    }{
    \mu(\sigma_{\mathrm{eff}}^2)^{2\mu}
    }
    \text{s}\left(
    \theta_k+\frac{3\mu\pi}{2}
    \right),
    \end{align}
where $\sigma_{\text{eff}}^2 = \sigma'^2-\bar{\Xi}_{\text{CLS}}$. Finally, averaging over the receiver orientation according to (27), and the orientation-dependent terms are approximated by their statistical averages to obtain the approximate outage probability of the CLS scheme.\vspace{-1mm}

\vspace{-2mm}
\subsection{Proof of Proposition \ref{Proposition6}}\label{Appendix8}\vspace{-1mm}
To derive $\mathcal{F}(-jt,\mathcal{A}_\mathbf{y})$, we substitute \eqref{eq13} into \eqref{eq_Pout}, and expressing the $e_{RT,x,i,0}$, $e_{RT,y,i,0}$, and $e_{RT,z,i,0}$ in terms of the cylindrical coordinates, the expression for $\mathcal{F}(-jt,\mathcal{A}_\mathbf{y})$ in VULO scheme can be expressed as
\begin{align}\label{lkj123}
    &\hspace{-2mm}\mathcal{F}(\hspace{-0.5mm}-jt,\mathcal{A}_\mathbf{y}\hspace{-0.7mm}) \hspace{-1.0mm}=\hspace{-1.7mm} \int_{0}^{2\pi} \hspace{-2.8mm}\int_{0}^{a_D}  \hspace{-5mm}\exp\hspace{-0.5mm}\Big\{\hspace{-1.5mm}-\hspace{-1mm}jt \gamma^{-1}\hspace{-0.9mm}n_l^{-2}\hspace{-0.8mm}\big(\text{c}\theta_{R}\sqrt{(n_l^2\hspace{-0.9mm}-\hspace{-0.9mm}1)(r^2\hspace{-1mm}+\hspace{-1mm}h_u^2)\hspace{-0.7mm}+\hspace{-0.7mm}h_u^2}\nonumber\\
    &\hspace{-1mm}+\hspace{-1mm}r\text{c}\theta_x\text{c}\theta_R\text{s}\phi_R \hspace{-1mm}+\hspace{-1mm}r\text{s}\theta_x\text{s}\theta_R\text{s}\phi_R\hspace{-0.5mm}\big)^{2}\hspace{-0.5mm}(r^2\hspace{-1mm}+\hspace{-0.5mm}h_u^2)^{-(m+3)}\hspace{-1mm}\Big\}rdrd\theta_x.
\end{align}
To simplify the above integral, the angular terms inside the exponential function can be grouped into $r\text{s}\phi_R\text{c}(\theta_x-\theta_R)$, the square-root term is approximated using the first-order Taylor series expansion, $\sqrt{(n_l^2-1)(r^2+h_u^2)+h_u^2}\approx n_lh_u\left(1+ \frac{(n_l^2-1)r^2}{2n_l^2h_u^2}\right)$. Then, we can approximate \eqref{lkj123} as
\begin{align}\label{dhdhd}
    &\mathcal{F}(\hspace{-0.9mm}-jt,\mathcal{A}_\mathbf{y}) \hspace{-0.8mm}\approx\hspace{-0.7mm} \exp\hspace{-0.7mm}\left(\hspace{-1.5mm}\frac{-jt \text{c}\theta_R h_u}{\gamma n_l}\hspace{-1.5mm}\right)\hspace{-2mm}\int_{0}^{2\pi} \hspace{-3mm}\int_{0}^{a_D} \hspace{-4mm} \exp\hspace{-0.2mm}\bigg\{\hspace{-0.8mm}\frac{-jt}{\gamma n_l}\hspace{-0.8mm}\bigg(\hspace{-0.8mm}\frac{\text{c}\theta_R(n_l^2\hspace{-0.8mm}-\hspace{-0.8mm}1)r^2}{2n_lh_u}\nonumber\\
    &+r\text{s}\phi_R\text{c}(\theta_x-\theta_R)\bigg)\hspace{-1mm}(r^2\hspace{-0.8mm}+\hspace{-0.8mm}h_u^2)^{-(m+3)}\hspace{-1mm}\bigg\}rdrd\theta_x\hspace{-0.4mm}.
\end{align}
Since the exponential argument is sufficiently small, the approximation $e^{x} \approx 1+x$ is employed. Therefore, we can write
\begin{align}\label{F_VULO_f43}
    &\mathcal{F}(\hspace{-0.8mm}-jt,\mathcal{A}_\mathbf{y}\hspace{-0.8mm}) \hspace{-0.7mm}\approx \hspace{-0.7mm}\exp\hspace{-0.8mm}\left(\hspace{-0.8mm}\frac{-jt\text{c}\theta_R h_u\hspace{-0.8mm}}{\gamma n_l}\right)\hspace{-2mm}\int_{0}^{2\pi} \hspace{-3mm}\int_{0}^{a_D} \hspace{-2.2mm}\bigg\{1\hspace{-1mm}- \hspace{-1mm}\frac{jt}{\gamma n_l}\bigg(\hspace{-0.7mm}\frac{\text{c}\theta_R(n_l^2\hspace{-0.7mm}-\hspace{-0.7mm}1)r^2}{2n_lh_u}\nonumber\\
    &+r\text{s}\phi_R\text{c}(\theta_x-\theta_R)\bigg)\hspace{-0.6mm}(r^2+h_u^2)^{-(m+3)}\bigg\}rdrd\theta_x\hspace{-0.2mm}.
\end{align}
Note that $\int_{0}^{2\pi}\text{c}(\theta_x-\theta_R)\,d\theta_x=0$. Using the substitution, $k = (r^2+h_u^2)$, the integral in~\eqref{F_VULO_f43} can be evaluated analytically, leading to the approximate closed-form expression.

\subsection{Proof of Proposition \ref{Proposition7}}\label{Appendix9}
$\mathcal{K}(-jt,\bar{\mathcal{A}}_{\mathbf{y}})$ in the VULO scheme can be expressed in cylindrical coordinates as
\begin{align}
    &\mathcal{K}(-jt,\bar{\mathcal{A}}_{\mathbf{y}}) \hspace{-0.5mm}=\hspace{-0.5mm} \text{exp}\bigg(\hspace{-1.8mm}-\hspace{-0.5mm}\lambda_m\hspace{-0.9mm}\int_{0}^{2\pi} \hspace{-2.5mm}\int_{R}^{a_D} \hspace{-1.9mm}\bigg(\hspace{-0.5mm}1\hspace{-0.5mm}-\hspace{-0.5mm}\exp\hspace{-0.5mm}\bigg[\hspace{-0.9mm}-\hspace{-0.5mm}\frac{jt}{\gamma n_l}\hspace{-0.5mm}\bigg(\hspace{-0.8mm}r\text{c}\theta_x\text{c}\theta_R\text{s}\phi_R \nonumber\\
    &+r\text{s}\theta_x\text{s}\theta_R\text{s}\phi_R+\text{c}\theta_R\sqrt{(n_l^2-1)(r^2+h_u^2)+h_u^2}\bigg)^2\nonumber \\
    &\times(r^2+h_u^2)^{-(m+2)}\bigg)\bigg]rdrd\theta_x\bigg).
\end{align}
For the interference field, we use a first-order approximation, the identities
$\text{c}\theta_x\text{c}\theta_R
+
\text{s}\theta_x\text{s}\theta_R
=
\text{c}(\theta_x-\theta_R),
$ $
\int_{0}^{2\pi}\text{c}(\theta_x-\theta_R)d\theta_x=0$, $\int_{0}^{2\pi}\text{c}^2(\theta_x-\theta_R)d\theta_x=\pi$, and the substitution $k=r^2+h_u^2$ to evaluate the first two terms of the integral as
\begin{align}
    &\hspace{-1.2mm}\int_{R}^{a_D}\hspace{-3mm}
    r^3(r^2\hspace{-0.5mm}+\hspace{-0.5mm}h_u^2)^{-(m+2)}dr
    \hspace{-0.5mm}=\hspace{-0.5mm}
    \frac{1}{2m}\hspace{-0.8mm}
    \left(\hspace{-0.5mm}
    (R^2\hspace{-0.8mm}+\hspace{-0.8mm}h_u^2)^{-m}
    \hspace{-0.9mm}-\hspace{-0.7mm}(a_D^2\hspace{-0.5mm}+\hspace{-0.5mm}h_u^2)^{-m}
    \right)
    \nonumber\\
    &-\frac{h_u^2}{2(m+1)}
    \left(
    (R^2\hspace{-0.5mm}+\hspace{-0.5mm}h_u^2)^{-(m+1)}
    -(a_D^2\hspace{-0.5mm}+\hspace{-0.5mm}h_u^2)^{-(m+1)}
    \right).
\end{align}
The square-root term in the integrand can be expanded using a Taylor series around $r=0$, and after algebraic manipulation it can be written as 
\begin{align}\label{eqseries1}
    &(r^2+h_u^2)^{-(m+2)}
    \big((n_l^2-1)(r^2+h_u^2)+h_u^2\big)^{1/2}\\
    &\hspace{-0.5mm}=\hspace{-0.5mm}
    h_u^{-2(m+1)}\hspace{-0.5mm}
    \sum_{l=0}^{\infty}\hspace{-1.2mm}
    \Bigg[\hspace{-1.2mm}
    \sum_{k=0}^{l}
    (-1)^l\hspace{-0.5mm}
    \frac{1}{h_u^{2l}}\hspace{-1.3mm}
    \binom{m\hspace{-0.5mm}+\hspace{-0.5mm}2}{k}\hspace{-1.5mm}
    \binom{1}{l\hspace{-0.5mm}-\hspace{-0.5mm}k}\hspace{-1.5mm}
    \left(
    \frac{n_l^2\hspace{-0.5mm}-\hspace{-0.5mm}1}{n_l^2}
    \right)^{\hspace{-1.5mm}l-k}\hspace{-1.2mm}
    \Bigg]\hspace{-0.5mm}
    r^{2l}\hspace{-1.5mm}.\nonumber
\end{align}
Using \eqref{eqseries1}, the final expression for $\mathcal{K}(-jt,\bar{\mathcal{A}}_{\mathbf{y}})$ is derived.

\subsection{Proof of Theorem \ref{Theorem3}}\label{Appendix10}
We substitute~\eqref{F_VULO_final} and~\eqref{K_x} into~\eqref{eq_Pout} to obtain $P_o^{VULO}(\gamma,\mathbf{y})$ in the VULO scheme. Then, with some mathematical manipulations, we obtain
\begin{align}
    &P_o^{\mathrm{VULO}}(\gamma,\mathbf y)
    \approx
    1-\frac{\lambda_m|\mathcal A_{\mathbf y}|}{2}
    \nonumber\\
    &+
    \frac{\lambda_m}{\pi}
    \int_0^\infty
    \frac{
    \pi a_D^2
    \text{s}(t\sigma_{\mathrm{eff},V}^2)
    +
    t\Xi_F
    \text{c}(t\sigma_{\mathrm{eff},V}^2)
    }{t}
    dt,
\end{align}
where $\sigma_{\mathrm{eff},V}^{2} = \sigma'^2-\Xi_{\mathrm{VULO}}$. Here $\Xi_{\mathrm{VULO}}$ and $\Xi_{\mathrm{F}}$ are functions of the
receiver orientation angles $(\theta_R,\phi_R)$ only and do not depend
on the integration variable $t$. Next, we use the series expansion of $\text{s}\theta$ and $\text{c}\theta$ to obtain 
  \begin{align}
    P_o^{\mathrm{VULO}}(\gamma,\mathbf y)
    &\hspace{-0.5mm}\approx\hspace{-0.5mm}
    1\hspace{-0.7mm}-\hspace{-0.7mm}\frac{\lambda_m |\mathcal A_y|}{2}\hspace{-0.7mm}+\hspace{-0.7mm}
    \frac{\lambda_m}{\pi}\hspace{-1mm}
    \sum_{k=0}^{K_{\max}}
    (-1)^k\hspace{-0.5mm}
    \bigg[
    \frac{\hspace{-0.5mm}
    \pi a_D^2\hspace{-0.5mm}
    \big(\hspace{-0.5mm}
    \sigma_{\mathrm{eff},V}^{2}\hspace{-0.5mm}
    \big)^{2k+1}
    }{
    (2k+1)\epsilon^{2k+1}
    }
    \nonumber \\
    &+
    \frac{
    \Xi_F
    \big(
    \sigma_{\mathrm{eff},V}^{2}
    \big)^{2k}
    }{
    \epsilon^{2k+1}
    }
    \bigg].
    \end{align}
Finally, using the definitions of $\sigma_{\mathrm{eff},V}^{2}$, $\Xi_{\mathrm{VULO}}$, and $\Xi_{\mathrm{F}}$ and averaging over receiver orientation angles, we obtain $P_o^{\mathrm{VULO}}(\gamma)$.

\bibliographystyle{IEEEtran}
\bibliography{IEEEabrv,Main}

\end{document}